\documentclass[aps, prx, amssymb, longbibliography, singlecolumn, superscriptaddress, 10pt, nofootinbib]{revtex4-2}
\usepackage[utf8]{inputenc}
\usepackage{background}
\usepackage[normalem]{ulem}
\usepackage{comment}
\usepackage[linewidth=1pt]{mdframed}
\usepackage{graphicx} 
\usepackage{dcolumn} 
\usepackage{bm} 
\usepackage{braket}
\usepackage{amsthm}
\usepackage{amsmath}
\usepackage{float}

\backgroundsetup{
   scale=1,
   angle=0,
   opacity=1,
   color=black,
   contents={\begin{tikzpicture}[remember picture, overlay, font=\sffamily]
     \end{tikzpicture}}
 }

\usepackage{listings}
\usepackage{xcolor}
\usepackage{tikz}
\usepackage{hyperref}
\hypersetup{                    
	colorlinks,
	citecolor=blue,
	filecolor=blue,
	linkcolor=blue,
	urlcolor=blue
}

\usepackage{orcidlink}

\usepackage{mathtools}
\usepackage[ruled,vlined]{algorithm2e}
\usepackage{amssymb}
\usepackage[ampersand]{easylist}
\usepackage{enumitem}
\hypersetup{
  colorlinks = true, 
  linkcolor={blue!50!black}, 
  citecolor=red, 
  urlcolor=blue, 
  linktoc=page,
}
\usepackage[capitalise, nameinlink]{cleveref}
\usepackage[
  style=long,
  nolist, 
  nonumberlist,
  nopostdot,
  acronym,
  toc=true, 
]{glossaries}
\usepackage{bm}
\usepackage{booktabs}
\usepackage{multirow}

\usepackage{array}
\usepackage[british]{babel}

\newtheorem{theorem}{Theorem}

\newtheorem{lemma}[theorem]{Lemma}

\newacronym{jw}{JW}{Jordan Wigner}
\newacronym{dft}{DFT}{density functional theory}
\newacronym{dfpt}{DFPT}{density functional perturbation theory}

\newacronym{dmc}{DMC}{diagrammatic Monte Carlo}
\newacronym{qc}{QC}{quantum computer}
\newacronym{epc}{EPC}{electron-phonon coupling}
\newacronym{sb}{SB}{standard binary}
\newacronym{bec}{BEC}{Bose-einstein condensate}
\newacronym{nisq}{NISQ}{noisy intermediate-scale quantum}
\newacronym{pqs}{PQS}{perturbative quantum simulation}
\newacronym{ed}{ED}{exact diagonalisation}
\newacronym{tds}{TDS}{time dynamics simulation}
\newacronym{qvm}{QVM}{quantum virtual machine}
\newacronym{pv}{PV}{photovoltaic}
\newacronym{hs}{HS}{Hubbard-Stratonovich}
\newacronym{iqp}{IQP}{instantaneous quantum polynomial}
\newacronym{mc}{MC}{Monte Carlo}
\newacronym{otoc}{OTOC}{out-of-time-order correlators}
\newacronym{mcmc}{MCMC}{Markov chain Monte Carlo}
\newacronym{qft}{QFT}{Quantum fourier transform}
\newacronym{ptb}{PTB}{Partial Trace over Bosons}
\newacronym{lcu}{LCU}{Linear Combination of Unitaries}
\newacronym{qsp}{QSP}{Quantum Signal Processing}
\newacronym{qsvt}{QSVT}{Quantum Singular Value Transform}
\newacronym{dd}{DD}{dynamical decoupling}
\newacronym{qmc}{QMC}{quantum Monte Carlo}
\newacronym{scp}{SCP}{self-consistent phonon}
\newacronym{qbm}{QBM}{quantum Brownian motion}
\newacronym{bmme}{BMME}{Born-Markov master equation}
\newacronym{flo}{FLO}{Fermionic Linear Optics}
\newacronym{ldos}{LDOS}{Local Density Of States}
\newacronym{fft}{FFT}{Fast Fourier Transform}
\newacronym{oqs}{OQS}{open quantum systems}

\begin{document}
\title{Quantum algorithm for dephasing of coupled systems: decoupling and IQP duality}

\author{Sabrina Yue Wang}
\email{sabrina@phasecraft.io}
\affiliation{Phasecraft Ltd.}
\author{Raul A. Santos}
\email{raul@phasecraft.io}
\affiliation{Phasecraft Ltd.}
\date{\today}

\begin{abstract}
Noise and decoherence are ubiquitous in the dynamics of quantum systems coupled to an external environment.
In the regime where environmental correlations decay rapidly, the evolution of a subsytem is well described by a Lindblad quantum master equation.

In this work, we introduce a quantum algorithm for simulating unital Lindbladian dynamics by sampling unitary quantum channels without extra ancillas. Using ancillary qubits we show that this algorithm allows approximating general Lindbladians as well. For interacting dephasing Lindbladians coupling two subsystems, we develop a decoupling scheme that reduces the circuit complexity of the simulation.
This is achieved by sampling from a time-correlated probability distribution—determined by the evolution of one subsystem, which specifies the stochastic circuit implemented on 
the complementary subsystem.

We demonstrate our approach by studying a model of bosons coupled to fermions via dephasing, which naturally arises from anharmonic effects in an electron–phonon system coupled to a bath.
Our method enables tracing out the bosonic degrees of freedom, reducing part of the dynamics to sampling an \gls{iqp} circuit.
The sampled bitstrings then define a corresponding fermionic problem, which in the non-interacting case can be solved efficiently classically. We comment on the computational complexity of this class of dissipative problems, using the known fact that sampling from \gls{iqp} circuits is believed to be difficult classically.
\end{abstract}

\maketitle

\section{Introduction}
 
 Historically, much effort has been devoted to develop efficient quantum algorithms for unitary evolution, culminating in optimal approaches \cite{Berry_2006,Childs2010,Berry_2014,Berry_2015,Low2017,Low_2019,Childs2021} as a function of the evolution time. 
This should come as no surprise as a natural application of quantum computers is to study the time evolution of quantum systems \cite{deutsch1985quantum,Somaroo_1999}, providing a clear path towards quantum advantage as the system size increases \cite{miessen2023quantum, PC_Quant_2025, PC_Google_2025, Quant_super_2025}.

On the other hand, the study of the equally relevant time evolution of open systems, i.e. systems that evolve in the presence of external (an often uncontrolled) degrees of freedom, has received less attention \cite{Li_2023}. In the limit of fast relaxation of the external system (often denoted as {\it the bath}) the time evolution of a subsystem is captured by a Lindblad master equation \cite{Gorini1975,Lindblad1976, Breuer2007} that determines the dynamics of the system and can in principle realise universal quantum computation \cite{Verstraete2009}. 
In this {\it Markovian} regime several quantum algorithms have been proposed to simulate open quantum dynamics \cite{Kliesch_2011,Childs:2016izk,cleve2019,Schlimgen2021,Schlimgen2022,Joo_2023,li2023,Suri2023twounitary,watad2023,Di_Bartolomeo_2024,Borras2025,Ding2024,pocrnic2025} with different requirements and assumptions. 

As quantum computers have yet to reach maturity, simpler albeit not optimal quantum algorithms usually perform better in practice \cite{Childs2018,Childs_2019,Bosse_2025} as their prefactors are usually smaller than the ones of theoretical optimal proposals. Mirroring this, it is natural to look for algorithms that lower the quantum resources of simulating  open quantum dynamics in concrete scenarios. Motivated by this quest, here we introduce and analyse a quantum simulation algorithm that approximates the time evolution of an open quantum system generated by a {\it unital} Lindbladian
\begin{align}\label{eq:unital}
    \mathcal{L}(\rho)=-i[H,\rho]+\sum_{i=1}^N \left( L_i \rho L_i - \frac{1}{2}\{L_i^2,\rho\}\right),
\end{align}
where $L_i$ is an hermitian but otherwise arbitrary operator. In general, a sufficient condition for a unital Lindbladian $\mathcal{L}$ is that $\sum_i [L_i, L^{\dagger}_i] =0$. We consider in this work the subset where $L_i$ are hermitian. A unital Lindbladian $\mathcal{L}$ has the identity operator as a fixed point i.e $\mathcal{L}(\mathbb{I})=0$ and generates the {\it unital quantum channel} $\mathcal{E}_t:=e^{t\mathcal {L}}$ that preserves the identity operator as $\mathcal{E}(\mathbb{I})=\mathbb{I}$.  We achieve this by approximating this unital quantum channel by a {\it mixed unitary channel} $\mathcal{E}(\rho):=\sum_kp_kU_k\rho U_k^\dagger$, where $0\leq p_k\leq 1$ and $\sum_kp_k=1$ i.e $\mathcal{E}$ is a convex combination of unitary channels \cite{Audenaert_2008}.  Due to their structure, mixed unitary channels can be interpreted as the expected channel obtained by applying a stochastic unitary channel sampled from a probability distribution $p_k$. This generates an efficient way of approximating them.

Unital Lindbladians cannot decrease the entropy, and their dissipative part tends to destroy quantum coherence. As the identity operator is a fixed point of the evolution, one may ask about the classical complexity of simulating them. After all, if the quantum coherences are being erased and the system is driven to the infinite temperature state, is a quantum computer even needed to simulate these systems? While this intuition is correct if the Lindbladian evolution has a unique steady state and the evolution time is sufficiently long, the dynamics is actually nontrivial in the presence of symmetries, or if the Lindbladian has other fixed points. For example, in the physics of fermions hopping on a lattice and subject to dephasing \cite{Medvedyeva_2016} it has been shown that the Lindbladian, although unital, is gapless in the thermodynamic limit, leading to nontrivial long time behaviour. Moreover, several open problems about the interplay of dephasing and coherent dynamics exist, ranging from its role in the many body localisation transition \cite{Longhi_2023,Žnidarič_2010,Fischer_2016,Medvedyeva2016_MBL,Levi_2016,Marcantoni_2022}, its effect in quantum transport \cite{Longhi_2023,Rebentrost_2009,Xhek_2021,Longhi_20241}, the breaking of topological order \cite{kiely2025phasetransitiontopologicalindex}, to the effect of dephasing in the quantum-classical transition \cite{Lonigro2022}.

Here  we argue for the general computational hardness of this approach from a different perspective. We study bipartite systems where the jump operator $L_i=A_iB_i$ acts nontrivially in subsystems $\mathcal{S}_A$ and $\mathcal{S}_B$. For this scenario we develop a decoupling scheme that allows one to trace out one of the subsystems. This result is highlighted in \cref{theo:Dcube}, and used to study the problem of electron-phonon coupling due to anharmonicity in \cref{sec:elec_phonon}. There we show explicitly that after tracing out the phonon degrees of freedom, the evolution in the fermionic sector is determined by sampling an IQP circuit, something that is expected to be difficult classically \cite{Bremner_2010,Bremner_2016}.
The decoupling scheme presented in \cref{lem:decoup} also allows the efficient systematic study of {\it non-Markovian} dissipation on a quantum computer, something that to our knowledge has not been discussed extensively.

\subsection{Relation to previous work}
Several algorithms for the quantum simulation of Lindbladian dynamics has been developed. In \cite{Kliesch_2011} the first algorithm for simulating Lindbladians by Trotterisation was discussed, leading to a $O(t^2/\epsilon)$ scaling of  gate complexity with evolution time $t$ and error $\epsilon$ (omitting factors related with the system size). In \cite{Childs:2016izk} the authors improve on this to $O(t^{3/2}/\sqrt{\epsilon})$. Assuming sparse quantum jump operators represented by an oracle, they also provide an algorithm with query complexity $O((t^2/\epsilon){\rm polylog}(t/\epsilon))$. This result was improved in \cite{cleve2019} where a performance of $O(t\, {\rm polylog}(t/\epsilon))$ is achieved using linear combination of unitaries.  A quantum algorithm achieving this same scaling but without compressed encoding has been discussed in \cite{li2023}. In \cite{Schlimgen2021,Schlimgen2022} an unravelling approach to the simulation of Lindbladian dynamics is presented, where non-unitary operators are approximated by sums of unitaries, although without a complexity analysis. 
The same technique of splitting a generic operator into two unitaries but using the quantum singular value transform is discussed in \cite{Suri2023twounitary}.
Refs. \cite{Ding2024} and \cite{pocrnic2025} provide different algorithms for embedding the non-unitary dynamics into a Hamiltonian evolution, the former using a numerical scheme to map the unravelled equation into a dilated Hamiltonian problem while the latter uses repeated interactions and measurement.
An approach to estimate the matrix elements needed to compute the Lindblad equation is presented in \cite{Joo_2023}.
On the classical side, novel algorithms based on unravelling the evolution equation \cite{Werner_2016} and sampling trajectories in a tensor network framework has been introduced in \cite{Classical_2025}. Moreover, dynamics under the unital Lindbladians we consider in this work can alternatively be framed as gradient flow~\cite{fcjk-qbwh}.

Motivated by the constraints of current quantum capabilities, cheaper practical algorithms for simulation have also been explored. In \cite{watad2023} a variational quantum algorithm has been proposed. It uses a parameterised circuit to time evolve a vectorised representation of the density matrix. In \cite{Di_Bartolomeo_2024} the authors show that any Lindblad dynamics can be simulated using one bath qubit, with repeated resets. In the same context of simplifying quantum algorithms, \cite{Borras2025} combines second-order Trotter formulas with randomisation to approximate Lindbladian evolution via sampling from a distribution over unitaries.

\subsubsection*{Our contribution}

In this work we explore sampling unitaries that implement unital channels, much in the same spirit as in \cite{Borras2025}, where no ancilla qubits are needed, but where the type of non-unitary evolutions are constrained to mixed unitary channels. We depart from \cite{Borras2025} in two respects:
\begin{itemize}
    \item We concentrate on the generator \cref{eq:unital} and provide an explicit mixed-unitary channel that approximates it. The gate complexity for an approximation of the evolution up to time $t$ and precision $\epsilon$ scales as to $O(t^2/\epsilon)$
    \item We introduce the decoupling \cref{lem:decoup} that allows the study of coupled systems (as in \cref{theo:Dcube}) by reducing the evolution of the coupled system to two different Hamiltonian simulation problems where the bitstring output of one determines the circuit in the other and the expectation of this process generates the full dynamics.
\end{itemize}

We exemplify this algorithm by studying a model of electrons coupled to phonons. Tracing the bosonic degrees of freedom corresponding to the phonons, we find that the evolution in the fermionic subsystem is determined by a mixed unitary channel with a {\it time dependent} probability distribution that corresponds to the measurement outcomes of an IQP circuit. 
This non-Markovian evolution of the fermions is further investigated by sampling (in emulation) the IQP circuit over 10 qubits, allowing us to fully simulate the system classically over 10 trotter steps.

The paper is organised as follows, in \cref{sec:main_res} we motivate and present the main results. In \cref{sec:applic} we discuss some applications of the main results in concrete settings. In \cref{sec:disorder} we highlight the connection with disordered systems while in \cref{sec:elec_phonon} we use these results to study the effect of electron-phonon coupling on a physical system. By exactly tracing out the bosonic degrees of freedom, we obtain an effective description of the system in terms of a mixed unitary channel on the fermionic sector, with probabilities associated with an IQP circuit with a number of qubits proportional to the number of Trotter steps. We study the properties of the emergent IQP circuit and discuss the full quantum algorithm in this context.
We show the result of the time evolution of the electron-phonon system in \cref{sec:numerics}, where we observe the effect of the dissipation on the fermionic system induced by the coupling through dephasing with the bosonic bath. Finally
in \cref{sec:conclusion} we discuss some broader implications of our work in terms of the hardness of classical simulation of unital Lindbladian evolution.

\section{Main results}
\label{sec:main_res}
\subsection{Motivation}
Let's consider the unital Lindbladian with Hamiltonian $H$ and quantum jump operator $L=L^\dagger$ 
\begin{align}\label{eq:lindbladian}
    \mathcal{L}(\rho)=-i[H,\rho] + L \rho L - \frac{1}{2}\{L^2,\rho\},
\end{align}
where the first term captures the coherent evolution of the system while the second describes dephasing induced by the environment. The dissipative evolution $e^{t\mathcal{L}}$ generated by \cref{eq:lindbladian} can be approximated by 
\begin{align}\label{eq:map}
    \mathcal{E}(\rho):=e^{-iHt}\frac{1}{2}(e^{i{\sqrt{t}}L}\rho e^{-i{\sqrt{t}}L}+e^{-i{\sqrt{t}}L}\rho e^{i\sqrt{t}{L}})e^{iHt}=e^{t\mathcal{L}}(\rho)+O(t^{2}),
\end{align}
as all the odd powers in $\sqrt{t}$ cancel. This expression can be interpreted as the expectation of an stochastic channel $\mathcal{U}_{s}(\rho):= e^{-iHt}e^{si\sqrt{t}L}\rho e^{-si\sqrt{t}L}e^{iHt}$, where the random variable $s=\{+,-\}$ is sampled uniformly. From this result the following theorem follows
\begin{theorem}[Dissipative Dynamics - simplified]\label{thm:1}
The time evolution with the unital Lindblad operator 
\begin{align}
    \mathcal{L}(\rho)=-i[H,\rho]+\sum_{i=1}^{N_L} \left( L_i \rho L_i - \frac{1}{2}\{L_i^2,\rho\}\right),
\end{align}
where $N_L$ is the number of jump operators in the dissipative component and $L_i$ is hermitian, can be approximated by the expectation value of the stochastic channel
\begin{align}\label{eq:mixed_unitary}
   \mathcal{U}_{\bm s}(\rho):= e^{-iHt} \left(\prod_{j=1}^{N_L} e^{s_ji\sqrt{t}L_j}\rho e^{-s_ji\sqrt{t}L_j}\right)e^{iHt},
\end{align}
where each $s_j=\pm 1$ in ${\bm s}=(s_1,\dots s_{N_L})$ is a uniformly distributed random variable with probability distribution $p_{\bm s}=\frac{1}{2^{N_L}}$. The error incurred in this approximation is 
\begin{align}
\left\|\mathbb{E}(\mathcal{U}_{\bm s})-e^{t\mathcal{L}}\right\|_\diamond=\left\|\sum_{{\bm s}=\{+,-\}^{N_L}}p_{\bm s}\mathcal{U}_{\bm s}-e^{t\mathcal{L}}\right\|_\diamond= O(t^2),
\end{align}
for $t\in[0,1]$.
\end{theorem}
We prove this theorem in \cref{sec:proofs}. A implicit use of this connection has been used in the context of designing algorithms to approximate Gibbs states~\cite{ding2025end}. However we haven't found an explicit discussion in the literature. We make this connection explicit as it helps us build the subsequent theorems and lemmas.
Unital Lindbladians are defined as convex combinations of unitary channels so this connection is not extremely surprising. In the case of interacting unital Lindbladians, where the jump operator acts nontrivially in two subsystems, we can extend this result as follows.
Consider now a bipartite system $\mathcal{S}=\mathcal{S}_A\cup\mathcal{S}_B$ and the dephasing Lindbladian
\begin{align}
    \mathcal{L}(\rho)=AB\rho AB - \frac{1}{2}\{(AB)^2,\rho\},
\end{align}
where $A\in \mathcal{S}_A$ and $B\in\mathcal{S}_B$ are hermitian operators satisfying $[A,B]=0$ and $A^2=1$. For simplicity we do not consider any coherent evolution for now as it can be easily added back. This Lindbladian generates the evolution $e^{t\mathcal{L}}(\rho)$ and can as before, be approximated by the convex combination of unitary channels
\begin{align}
    \mathcal{E}(\rho):=\frac{1}{2}(e^{i{\sqrt{t}}BA}\rho e^{-i{\sqrt{t}}BA}+e^{-i{\sqrt{t}}BA}\rho e^{i\sqrt{t}{BA}})=e^{t\mathcal{L}}(\rho)+O(t^{2})
\end{align}
Using that $A^2=1$ we can decompose the channel $\mathcal{E}$ into the sum
\begin{align}
    \mathcal{E}(\rho)=\sum_{\gamma=0,1}\frac{A^{\gamma}}{2}\sum_{s^{+},s^{-}=\pm}e^{i\sqrt{t}Bs^{+}}\rho e^{-i\sqrt{t}Bs^{-}}(s^{+}s^{-})^{\gamma}\frac{A^{\gamma}}{2}.
\end{align}
This form motivates defining the state $|+\rangle=\frac{1}{\sqrt{2}}(|1\rangle+|-1\rangle)$ which is the $+1$ eigenstate of the Pauli $X$ operator, such that
\begin{align}
    \mathcal{E}(\rho)= \sum_{\gamma} A^{\gamma}\langle\gamma|e^{i\sqrt{t}BZ}|+\rangle \rho\langle +|e^{-i\sqrt{t}BZ}|\gamma\rangle A^{\gamma},
\end{align}
where $|\gamma\rangle:=Z^\gamma|+\rangle$ is the eigenstate of $X$ with eigenvalue $(-1)^\gamma$. This decomposition is very suggestive as it completely decouples the product of operators $A$ and $B$ at the cost of coupling $B$ with an extra spin (qubit) degree of freedom (see \cref{fig:decoupling}). 

\begin{figure}
    \centering
    \includegraphics[width=0.8\linewidth]{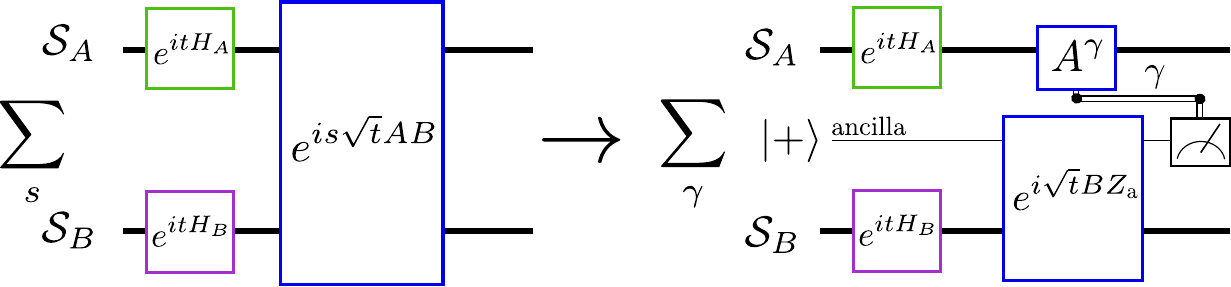}
    \caption{Decoupling gadget: The average of a channel with interaction term $e^{is\sqrt{t}AB}$ where $s=\{-1,1\}$ can be expressed as the outcome of a different channel where one subsystem is coupled to an ancilla through the Pauli Z operator and then the ancilla is measured. The result of that measurement defines the circuit that is implemented in the other subsystem. This gadget is technically a feedforward mechanism in general, but it does not need to maintain coherence between the subsystems $S_A$ and $S_B$ as they can be \textit{completely} separated. Therefore, in practice, the full circuit can be exactly constructed via disconnected, smaller circuits through classical postprocessing.}
    \label{fig:decoupling}
\end{figure}

Taking the partial trace with respect to the subsystem $\mathcal{S}_B$ on a separable density matrix $\rho = \rho_{\mathcal{S}_A}\otimes \rho_{\mathcal{S}_B}$ we have
\begin{align}\label{eq:split}
    {\rm Tr}_{\mathcal{S}_B}(\mathcal{E}(\rho_{\mathcal{S}_A}\otimes \rho_{\mathcal{S}_B}))&=\sum_\gamma p_\gamma A^\gamma \rho_{\mathcal{S}_A}A^\gamma,
\end{align}
where $p_\gamma$ corresponds to the probability of measuring the state $|\gamma\rangle$ on the ancilla after evolving the initial decoupled density matrix between the ancilla and the system, i.e.
\begin{align}
    p_\gamma&:={\rm Tr}(e^{i\sqrt{t}BZ}\rho_{aux}\otimes \rho_{\mathcal{S}_B} e^{-i\sqrt{t}BZ}|\gamma\rangle\langle \gamma |)
\end{align}
where $\rho_{aux}=|+\rangle\langle+|$. What have we gained with this? From \cref{eq:split}, we see that to simulate the original dissipative evolution, we can first simulate the system $\mathcal{S}_B$ coupled to the ancillary qubit and based on the outcome of measuring the ancillary qubit in the $\gamma$ basis perform a unitary evolution in $\mathcal{S}_A$. This decoupling mechanism is always beneficial when the coupling appears just in the dissipative term compared with the direct simulation of the entire system as the ancillary system does not have its own dynamics. Adding back the the unitary evolution and several types of couplings we find the result
\begin{theorem}[Dissipative Decoupled Dynamics: Dcube]\label{theo:Dcube}
Given a system $\mathcal{S}=\mathcal{S}_A\cup\mathcal{S}_B$, a Hamiltonian $H=H_A+H_B$ where $H_{A/B}$ acts only in its respective subspace, and hermitian jump operators $L_i=A_iB_i$ with $[A_i,B_j]=0$ and $A_j^2=1$ $\forall j$, where $A_i (B_i)$ acts nontrivially on $S_A(\mathcal{S}_B)$, the time evolution generated by the unital Lindbladian 
\begin{align}
    \mathcal{L}(\rho)=-i[H,\rho]+\sum_{i=1}^{N_L} \left( L_i \rho L_i - \frac{1}{2}\{L_i^2,\rho\}\right),
\end{align}
over a time $t$ can be approximated by introducing an ancillary qubit system $\mathcal{S}_C$ such that
\begin{align}
  e^{t\mathcal{L}}(\rho) =  \sum_{\bm\gamma=\{0,1\}^{N_L}}U_{\bm\gamma}^{(A)}{\rm Tr}_C(V^{(B,C)}\rho\otimes\rho_C V^{\dagger(B,C)}\Pi_{\bm\gamma})U_{\bm\gamma}^{\dagger(A)} +O(t^2)
\end{align}
with $U_{\bm\gamma}^{(A)}:=e^{-itH_A}\prod_{j=1}^{N_L}A_j^{\gamma_j}$, $V^{(B,C)}:=e^{-itH_B}\prod_{j=1}^{N_L}e^{i\sqrt{t}B_jZ_j}$, $\rho_C=\otimes_{i=1}^{N_L}|+\rangle\langle +|_i$ and $\Pi_{\bm\gamma}$ is a projector onto the state $|\bm\gamma\rangle :=|\gamma_1,\dots,\gamma_{N_L}\rangle$ where $|\gamma\rangle$ is the eigenstate of the Pauli X matrix with eigenvalue $(-1)^\gamma$.
\end{theorem}
The usefulness of this result is highlighted in the following lemma
\begin{lemma}[Decoupling lemma]\label{lem:decoup}
Assuming the same conditions as in \cref{theo:Dcube}; the time evolution of an initially unentangled density matrix between the subsystems $\rho_0=\rho_A\otimes\rho_B$ on the system $\mathcal{S}_A$ can be approximated as
\begin{align}\label{eqn:decoup_trace}
    \rho_A(t)={\rm Tr}_{\mathcal{S}_B}\left(e^{t\mathcal{L}}(\rho_0)\right)=\sum_{\bm\gamma=\{0,1\}^{N_LR}}{\rm Tr}({V}_t(\rho_B\otimes\Pi_{\bm 0})V_t^\dagger\Pi_{\bm\gamma}){U}_{\bm\gamma,t}(\rho_A){U}^\dagger_{\bm\gamma,t}+O(t^2/R)
\end{align}
where ${U}_{\bm\gamma,t}:=\prod_{k=1}^R \left(e^{-i\frac{t}{R}H_A}\prod_{j=1}^{N_L}A_j^{\gamma_{jk}}\right)$ and $V_t:=\prod_{j=1}^R\left(e^{-i\frac{t}{R}H_B}\prod_{k=1}^{N_L}e^{i\sqrt{\frac{t}{R}}B_kZ_{j,k}}\right)$. Here $\Pi_{\bm\gamma}$ is a projector onto the state $|\bm\gamma\rangle:=\otimes_{j,k}|\gamma_{j,k}\rangle$ where $|\gamma\rangle$ is the eigenstate of the Pauli X matrix with eigenvalue $(-1)^\gamma$. $Z$ denotes the conventional Pauli matrix. 
\end{lemma}

In the most general sense, $B_k$ can also have a $j$ (Trotter) index. This means that to simulate the time evolution on $\mathcal{S}_A$ we can stochastically sample unitary evolutions defined only on $\mathcal{S}_A$ according to a probability distribution that is governed by the evolution of $\mathcal{S}_B$ coupled to a number of qubits that grow with the number of Trotter steps.

The different algorithms for \gls{tds} of open systems in the context of \cref{thm:1} and \cref{theo:Dcube} are presented below as pseudocode in Algorithms \ref{alg:pseudocode_1} and \ref{alg:pseudocode_2}. 

\subsection{General connection of Lindbladian dynamics to \gls{iqp} circuits}
In this section, we discuss the general form of the circuits that appear on the ancillary system.
Consider the trace operation over subsystem $\mathcal{S}_{B}$ and the ancillary qubit system $\mathcal{S}_C$ in \cref{eqn:decoup_trace}, i.e.
\begin{equation}
    \centering
    \Omega_{\bm{\gamma}} = \text{Tr}\left(V_t (\rho_B \otimes  \rho_C) V_t^{\dagger} \Pi_{\bm{\gamma}}\right)
\end{equation}
Using the spectral decomposition of the operator $B_k=\sum_{\lambda}\lambda \Pi_{k,\lambda}$ the unitary $V_t$ becomes
\begin{align}
    V_{t}&=\sum_{\bm{\lambda}}\prod_{j=1}^{R}\left(e^{-i\frac{t}{R}H_{B}}\prod_{k=1}^{N_{L}}\Pi_{k,\lambda_{j,k}}^{B}e^{i\lambda_{j,k}\sqrt{t/N}Z_{j,k}}\right)
\end{align}
where the sum is over all possible values of the eigenvalues of each of the projector operators appearing in the product, such that 
\begin{align}
\Omega_{\bm\gamma}=\sum_{\bm\lambda\bm\lambda'}{\rm Tr}(O_{\bm\lambda}\rho_B O_{\bm\lambda'}^\dagger){\rm Tr}(\mathcal{O}_{\bm\lambda}\rho_C \mathcal{O}_{\bm\lambda'}^\dagger)
\end{align}
where
\begin{align}
    O_{\bm\lambda} := \prod_{j=1}^{R}\left(e^{-i\frac{t}{R}H_{B}}\prod_{k=1}^{N_{L}}\Pi_{k,\lambda_{j,k}}^{B}\right)\quad \mbox{and} \quad 
    \mathcal{O}_{\bm\lambda} := e^{i\sum_{j,k}\lambda_{j,k}\sqrt{t/N}Z_{j,k}}
\end{align}

$\mathcal{O}_{\bm{x}}$ and $\mathcal{O}_{\bm{y}}$ only consists of $R_z$ rotations, and recall that the ancillas are initialised in the $+1$ eigenstate of the $X$ basis, $\rho_C:=\otimes_{i=1}^{N_LR}|+\rangle\langle +|_i = \ket{\oplus} \bra{\oplus}$. Therefore 
\begin{align}
    \centering
    \Omega_{\bm{\gamma}} &= \sum_{\bm{\lambda},\bm{\lambda}'} \text{Tr}\left( O_{\bm{x}} \rho_B O_{\bm{y}}  \right) \bra{\bm{\gamma}} \mathcal{O}_{\bm{x}} \ket{\oplus} \bra{\oplus} \mathcal{O}_{\bm{y}} \ket{\bm{\gamma}}
\end{align}
where the action on the ancilla subsystem generates an \gls{iqp} circuit. In the case $\bm{x}=\bm{y}$, the ancilla part is exactly an \gls{iqp} probability, such that $\Omega_{\bm\gamma}$ can be interpreted as a Monte carlo sum where we sample from an \gls{iqp} circuit. Ultimately, the operator $B_k$ can only introduce correlations in the ancilla system of the form $e^{i\theta \prod_{l=1}^J Z_{j_l}}$ for $J\leq N_L R-1$ and some qubit index $j_l$. Therefore $\Omega_{\gamma}$ will always be a CPTP map consisting of only commuting gates on the ancilla system. An explicit example is discussed in the application to electron-phonon systems \cref{eqn:Gamma_explicit}, where the explicit trace over the subsystem $B$ generates an emergent \gls{iqp} circuit made of exponentials of single and quadratic $Z$ monomials.

\subsection{Extension to general Lindbladians}

Up to this point, we have concentrated in the approximation of unital Lindbladians, which, by virtue of \cref{thm:1} are implementable sampling unitary channels. Using ancillas, it is possible to generalise this arbitrary Lindbladians. Consider the map \cref{eq:map} with the hermitian Lindblad operator $L=(\ell\sigma_a^++\ell^\dagger \sigma_a^-)$. 
\begin{align}
    \mathcal{E}(\rho):=e^{-iHt}\frac{1}{2}\left(e^{i{\sqrt{t}}(\ell\sigma_a^++\ell^\dagger \sigma_a^-)}\rho e^{-i{\sqrt{t}}(\ell\sigma_a^++\ell^\dagger \sigma_a^-)}+e^{-i{\sqrt{t}}(\ell\sigma_a^++\ell^\dagger \sigma_a^-)}\rho e^{i\sqrt{t}(\ell\sigma_a^++\ell^\dagger \sigma_a^-)}\right)e^{iHt},
\end{align}
here $\ell$ is some arbitrary operator on the target system $\rho_s$ and $\sigma_a^\pm=\frac{1}{2}(X_a\pm i Y_a)$ is the common raising/lowering operator acting on an ancillary qubit. Applying this map to the density matrix $\rho = \rho_s\otimes (|0\rangle \langle 0|)_a$ using \cref{thm:1} and tracing away the ancillary qubit generates the approximation
\begin{align}
    {\rm Tr}_a(\mathcal {E}(\rho_s(|0\rangle\langle 0 |)_a))=e^{t\mathcal{L}}(\rho_s)+O(t^2)
\end{align}
with
\begin{align}
    \mathcal{L}(\rho_s)=-i[H,\rho_s]+\ell \rho_s \ell^\dagger - \frac{1}{2}\{\ell^\dagger \ell, \rho_s\}.
\end{align}
Adding more ancillas this can be trivially generalised to the case of different jump operators.

\begin{algorithm}

\SetAlgoLined
\KwResult{An approximation of the expectation of a time evolved operator $\hat{O}(t)$ under the Lindbladian
$\mathcal{L}(\rho)=-i[H,\rho]+\sum_{j=1}^N L_j \rho L_j -\frac{1}{2}\{L_j^2,\rho\}$}
 Given a Hamiltonian $H=\sum_{k=1}^l h_k$ and a set of quantum jump operaror $L_j$ with $j\in\{1,\dots,N_L\}$ where $h_k$ and $L_j$ are operators such that $e^{ith_k}$ and $e^{i\sqrt{t}L_j}$ are directly implementable, a state $\rho_0$, an operator $\hat{O}$, an evolution time $t$ a number of Trotter steps $R$ and a number of measurements $M$,

\medskip
\For{$p=1..M$}{
\begin{itemize}

\item Sample the $N_L\times R$ random variables $\bm s := \{\{s_{jk}\}_{j=1}^{N_L}\}_{k=1}^R$ from the uniform distribution where\\ $s_{jk}=\{-1,1\}$.

\item Construct the unitary
\begin{align*}
{U}_{\bm s}:= \prod_{r=1}^R\left(\prod_{k=1}^le^{-i\frac{t}{R}h_k}\prod_{j=1}^{N_L} e^{s_{jr}i\sqrt{\frac{t}{R}}L_j}\right)
\end{align*}

\item Measure $x_p:={\rm Tr}(U_{\bm s}\rho_0 U_{\bm s}^\dagger O)$
\end{itemize}
}
Output $\frac{1}{M}\sum_{p=1}^M x_p$

 \caption{Approximating the time evolution of systems with dephasing}\label{alg:pseudocode_1}
\end{algorithm}

\begin{algorithm}

\SetAlgoLined
\KwResult{An approximation of the expectation of a time evolved operator $\hat{O}(t)$ under Lindbladian evolution
$\mathcal{L}(\rho)=-i[H,\rho]+\sum_{i=1}^{N_L}A_jB_j\rho B_jA_j-\frac{1}{2}\{B_j^2,\rho\}$ with $A_j^2=1$ and $[B_j,A_k]=0$ $\forall j,k$}
 Given a Hamiltonian $H=H_A+H_B$ with $H_{\alpha}$ acting on subsystem $\mathcal{S}_{\alpha}$ and $H_{\alpha}=\sum_{k=1}^l h_{\alpha,k}$ and a set of quantum jump operator $A_jB_j$ with $j\in\{1,\dots,N_L\}$ where $h_{\alpha,k}$ and $B_j$ are operators such that $e^{ith_{\alpha,k}}$ and $e^{i\sqrt{t}B_jZ_a}$ are directly implementable, with $Z_a$ a Pauli $Z$ operator acting on an ancilla qubit; a state $\rho_0$, an operator $\hat{O}$, an evolution time $t$ a number of Trotter steps $R$ and a number of measurements $M$

\medskip
\For{$p=1..M$}{
\begin{itemize}
\item Construct the state $\rho_C:=\otimes_{i=1}^{N_LR}|+\rangle\langle +|_i$ using $N_LR$ ancillas
\item Define the unitary
\begin{align*}
{V}:=\prod_{j=1}^R\left(\prod_{k=1}^le^{-i\frac{t}{R}h_{B,k}}\prod_{k=1}^{N_L}e^{i\sqrt{\frac{t}{R}}B_kZ_{j,k}}\right)
\end{align*}
and the circuit $V\rho_0\otimes\rho_C V^\dagger$
\item Measure the ancillas in the $X$ Pauli basis. Store the resulting bitstring as ${\bm\gamma}_p$. \\
The remaining state is $\tilde \rho$

\item Define the unitary ${U}_{\bm\gamma}:=\prod_{k=1}^R \left(\prod_{k=1}^le^{-i\frac{t}{R}h_{A,k}}\prod_{j=1}^{N_L}A_j^{\gamma_{jk}}\right)$ and the circuit $U_{\bm \gamma}\tilde\rho U_{\bm \gamma}^\dagger$.
\item Measure $x_p:={\rm Tr}(U_{\bm \gamma }\rho_0 U_{\bm \gamma}^\dagger \hat{O})$
\end{itemize}
}
Output $\frac{1}{M}\sum_{p=1}^M x_p$

\caption{Approximating the time evolution of systems interacting through dephasing}\label{alg:pseudocode_2}
\end{algorithm}

In the next section we discuss applications of these algorithms, in particular the connection with the study of disordered systems \cref{sec:disorder} and the physics of electron-phonon coupled systems  \cref{sec:elec_phonon}.

\section{Applications}\label{sec:applic}

\subsection{Connection with simulation of disordered systems}
\label{sec:disorder}
\cref{thm:1} can be used to show that dephasing appears naturally in the study of disordered systems, where the value of an observable corresponds to the ensemble average over some disorder distribution. Here we show explicitcly that for particular disordered systems, the average over disorder corresponds to a Lindbladian evolution with dephasing and rescaled parameters. 

More precisely, consider the task of simulating the evolution under a random Hamiltonian $H(\xi):=H_0+\xi H_1$, where $\xi$ is sampled according to some probability distribution $p(\xi)$ and the results are averaged over $p(\xi)$. 
Consider the following disorder average of a time evolved operator $O$ on the state $|\Psi\rangle$
\begin{align}
    \mathbb E\langle O(t)\rangle_\xi := \mathbb E(\langle \Psi| e^{itH(\xi)} O e^{-itH(\xi)}|\Psi\rangle):=\sum_\xi p(\xi)\langle \Psi| e^{itH(\xi)} O e^{-itH(\xi)}|\Psi\rangle
\end{align} 
where we assume that $\xi\in [a,b]$ is sampled uniformly. Here $H_{0,1}$ are Hamiltonians that are independent of the random variable $\xi$.
We can rescale $\xi$ in terms of a new random variable $\tilde{\xi}$ uniformly distributed in the interval [0,1] as $\tilde{\xi} :=\frac{\xi-a}{b-a}$. Using the binary representation of $\xi$ we have
\begin{align}
    \tilde{\xi}=\sum_{n=1}^\infty\frac{\xi_n}{2^n}\rightarrow\xi =(b-a)\sum_{n=1}^\infty\frac{\xi_n}{2^n}+a
\end{align}
where $\xi_n\in\{0,1\}$ is a discrete random variable. Finally, defining $s_n:=2\xi_n-1\in\{-1,1\}$ we have
\begin{align}
    H(\xi)=H_0+\xi H_1 = H_{0}+\frac{(b+a)H_{1}}{2}+\sum_{n=1}^{\infty}s_{n}\frac{(b-a)}{2^{n+1}}H_{1}
\end{align}
The first order Trotterisation of the evolution by this Hamiltonian is 
\begin{align}
    e^{itH(\xi)}=e^{it(H_0+\frac{a+b}{2}H_1)}\prod_{n=1}^\infty e^{s_ni\sqrt{t}\left(\sqrt{t}(b-a)\frac{H_1}{2^{n+1}}\right)}+O(t^2)
\end{align}
Using \cref{thm:1} we find that the average channel $\mathbb E(U_\xi \rho U_\xi^\dagger)=\sum_\xi p(\xi)e^{itH(\xi)}\rho e^{-itH(\xi)}$ is approximated by the Lindbladian
\begin{align}
    \mathcal{L}(\rho)&= -i[H_0+\frac{a+b}{2}H_1,\rho]+\frac{t(b-a)^2}{12}\left(H_1\rho H_1 - \frac{1}{2}\{H_1^2,\rho\}\right)
\end{align}
with an error $O(t^2)$ as long as $b-a\sim~O(t^{-\frac{1}{2}})$.
This implies that the disorder average with a disorder sampled from an interval that decreases like $t^{-\frac{1}{2}}$ can be approximated by a Lindbladian evolution where the dephasing is induced by the disordered term. For other disorders, the dephasing process is not captured by a time independent and local Lindbladian, but can still be captured by a mixed unital quantum channel, of which a generator in time can be defined.

\subsection{Electron-phonon systems}
\label{sec:elec_phonon}

The main application of \cref{theo:Dcube} that we consider in this work is the dynamics of electrons coupled to phonons, where the coupling appears in the form of a quantum jump operator $\ell_i = n_i x_i$ composed of the electron density $n_i$ and the position operator of the boson around the lattice site $i$, $x_i$.
This quantum jump operator induces coupled dissipative dynamics on an electronic system representing electrons on a lattice and a bosonic system representing phonon modes related to the vibrations of the lattice around its equilibrium position. As it is common, we consider only one phononic mode per lattice site, corresponding to the coupling with the most relevant optical phonon mode.
These type of interactions have been considered in the context of pure Hamiltonian description in the Holstein~\cite{holstein1959studies} and Fröhlich Hamiltonians~\cite{frohlich1954electrons,tempere2009feynman}. The former includes localised electron-phonon coupling and is able to capture the emergence of polaron quasiparticles (particles where the charge is dressed by lattice distortions, creating a net (multi)pole). In contrast the latter includes long-range interactions and serves to understand symmetry breaking phenomena like BCS-superconductivity~\cite{bardeen1957theory} and the formation of charge-density waves~\cite{kvande2023enhancement}. 
These models often work in the harmonic approximation, where the purely bosonic Hamiltonian is reduced to an harmonic potential. Moreover, in this context the \gls{epc} terms are also often limited to a single electron coupled to a single phonon. 

There is a rich history of methods developed to study electron-phonon models, from Feynman's variational path-integral formalism~\cite{feynman1955slow} to \gls{dmc}~\cite{prokof2008bold,mishchenko2014diagrammatic,hahn2018diagrammatic}. Feynman's approach is a semi-analytical theory, where all bosonic degrees of freedom are traced away and their effects incorporated in fermionic self-interactions. This theory has been shown to fail in  dynamical scenarios \cite{sels2016dynamic} and intermediate coupling regimes~\cite{vlietinck2015diagrammatic}. To address dynamics,  \gls{dmc} performs a Wick rotation to map the problem into imaginary time, avoiding large oscillations in real time evolution. As a result analytical continuation is needed to recover back the real time signal. This procedure is in practice ill-defined and without theoretical guarantees of convergence. Recent works have attempted to circumvent this issue by using symbolic analytical continuation~\cite{taheridehkordi2019algorithmic,vuvcivcevic2020real}, which better controls the rotation back to real time, but also lack guarantees in general.



The accurate treatment of anharmonic effects is computationally challenging~\cite[\& references therein]{houtput2024anharmonic} but essential for understanding real materials. For example, in a material approaching a structural phase transition some vibration modes will have a frequency going to zero in the harmonic case. Including non-harmonic contributions controls the type of structural transition that occurs ~\cite{cochran1960crystal}. Similarly, phonon-phonon interactions, a direct consequence of anharmonicity, greatly alter electron-phonon scattering rates in the study of transport~\cite{mellan2019electron}. 

Recently, several proposals to study these anharmonic effects have appeared in the literature, by extending analytic methods \cite{houtput2024anharmonic}, self-consistently through \gls{scp} theory~\cite{hooton1955li, werthamer1970self, zacharias2023anharmonic}, or using coupled cluster~\cite{white2020coupled} and \gls{qmc}~\cite{paleari2021quantum}.

In this section we study anharmonic effects in the dynamical evolution of an electron-phonon system by analysing a system composed of subsystems $F$ which is fermionic, $Q$ and $B$ bosonic, each with internal dynamics described by the Hamiltonians $H_{F}$, $H_{Q}$ and $H_{B}$ respectively.  We assume the following anharmonic coupling
\begin{align}
    H_c =  \frac{1}{2}\sum_{i=1}^L g_i (q^{\dagger}_i + q_i) (b^{\dagger}_i + b_i) (f^{\dagger}_i f_i - \frac{1}{2}),
\end{align}
where $q_i$ $(q_i^\dagger)$ is the bosonic destruction (creation) operator corresponding to a Brownian particle at site $i$, $b^\dagger_i$ and $b_i$ are bosonic creation and destruction operators of an external bath, and $f_i$ $(f_i^\dagger)$ are fermionic destruction (creation) operators, all defined at the same lattice site. Assuming an harmonic bath, it is possible to trace the effect of the bath by expanding the von Neumann equation to second order in the anharmonic interaction. This procedure (see \cref{appendix:epi_lindbladian_derivation} for details), generates the following Lindbladian


\begin{align}\label{eq:Lind_eph}
\mathcal{L}(\rho) &=-i[H_{F}+H_{Q},\rho]+\frac{1}{4}\sum_{j=1}^L g_j^2\mathcal{L}_{x_j(2n_j-1)}(\rho),
\\&= -i [H_{F}+H_{Q},\rho]+\sum_{j=1}^L g_j^2\left(x_{j}\mathcal{L}_{n_{j}}(\rho)x_{j}+\frac{1}{4}\mathcal{L}_{x_{j}}(\rho)\right),
\end{align}
with $\mathcal{L}_{\ell}(\rho):=\ell\rho(t)\ell^\dagger-\frac{1}{2}\{\ell^\dagger \ell,\rho\}$ and $x_j =\frac{q_j+q_j^\dagger}{\sqrt{2}}$. Here  the usual dephasing Lindbladian for electrons $\mathcal{L}_{n_j}(\rho)$ ~\cite{dolgirev2020non} is multiplied by the position of the quantum Brownian particle, producing an interaction dependent dephasing on the fermionic sector. $\mathcal{L}_x(\rho)$ is the dissipative term appearing in 
\gls{qbm}~\cite{huang2022exact}, a prevalent model describing systems coupled to bosonic baths.  Traditional methods to tackle the \gls{qbm} part using \gls{bmme}~\cite{redfield1957theory,davies1974markovian,wangsness1953dynamical,Lindblad1976} often violates the positivity condition of the density matrix due to the loss of Markovianity in the construction, and lead to the break-down of Heisenberg's uncertainty principle. Using a Lindbladian construction~\cite{lampo2016lindblad} these problems dissapear.

The Lindbladian \cref{eq:Lind_eph} satisfies the requirements of the decoupling \cref{lem:decoup} and as such the electron-phonon dynamics can be completely decoupled via the introduction of Ising-like variables (ancillary qubits) that mediate interactions between the subsystems $F$ and $Q$. This is very useful as it allows completely tracing the bosonic degrees of freedom. 
This result is captured in the following lemma
\begin{lemma}{(Boson - Qubit duality in dissipative evolution.)}\label{lem:boson_qubit}
    The fermionic density matrix $\rho_F(t)={\rm Tr}_Q(e^{t\mathcal{L}}\rho(0))$ of a system evolving under the Lindbladian \cref{eq:Lind_eph}
    \begin{align}\label{eqn:L_epi}
        \mathcal{L}(\rho) =-i[H_{F}+H_{Q},\rho]+\frac{1}{4}\sum_{j=1}^L g_j^2\mathcal{L}_{x_j(2n_j-1)}(\rho),
    \end{align}
    where $H_{Q}=\sum_{j=1}^L \omega_j q_j^\dagger q_j$ is a sum of harmonic Hamiltonians over all the bosons and the initial state is unentangled $\rho(0)=\rho_F\otimes|0\rangle \langle 0|$, can be approximated to arbitrary accuracy by the mixed unital channel
    \begin{align}\label{eq:rho_F_lem}
            \rho_{F}(t)=\sum_{\bm{\gamma}=\{0,1\}^{LN}}{\rm Tr}\left(e^{t\mathcal{L}_{{\rm IQP}}(t)}(\Pi_{\bm{0}})\Pi_{\bm{\gamma}}\right)U_{\bm{\gamma},t}(\rho_{F})U_{\bm{\gamma},t}^{\dagger}+O(t^{2}/N),
    \end{align}
    following the definitions of \cref{lem:decoup} and the effective Lindbladian on the ancilla qubits
    \begin{align}\label{eq:iqp}
        \mathcal{L}_{\rm IQP}(t)[\rho]:=-i[H_{\rm IQP}(t),\rho]+\sum_{k=1}^L \left(\ell_k(t)\rho\ell^\dagger_k(t)-\frac{1}{2}\left\{\ell^\dagger_k(t)\ell_k(t),\rho\right\}\right),
    \end{align}
where $H_{\rm IQP}(t):=\sum_{k=1}^L\frac {g_k^2}{8N}\sum_{j>l}Z_{j,k}Z_{l,k}\sin\left(\omega_k\frac{(l-j)t}{N}\right)$, $\ell_k(t)=\frac{g_{k}}{2}\sqrt{\frac{1}{2N}}\sum_{j=1}^{N}e^{i\omega_{k}\frac{jt}{R}}Z_{j,k}$ and $N$ is the number of Trotter steps.
\end{lemma}
\begin{proof}
 Directly applying \cref{lem:decoup} on the Lindbladian \cref{eq:Lind_eph}, we have
\begin{align}\label{eq:rho_eph}
        \rho_F(t)=\sum_{\bm\gamma=\{0,1\}^{LN}}{\rm Tr}\left(V_t(\rho_Q\otimes\Pi_{\bm 0})V_t^\dagger\Pi_{\bm\gamma}\right){U}_{\bm\gamma,t}(\rho_F){U}^\dagger_{\bm\gamma,t}+O(t^2/N),
\end{align}
where
\begin{align}\label{eq:Vt}
    V_t:=\prod_{j=1}^N\left(e^{-i\frac{t}{N}H_Q}\prod_{k=1}^Le^{i\sqrt{\frac{t}{N}}g_k\frac{x_k}{2}Z_{j,k}}\right) = e^{-itH_Q}\prod_{j=1}^N\left(\prod_{k=1}^L\exp\left[{\frac{ig_k}{2}\sqrt{\frac{t}{N}}{x_k}\left(\frac{jt}{N}\right)Z_{j,k}}\right]\right).
\end{align}

The last equality is an identity that follows from the repeated application of
\begin{align}
    e^{-tA}e^{-tB}e^{-tA}e^{-tB}=e^{-2tA}(e^{tA}e^{-tB}e^{-tA})(e^{-tB})=e^{-2tA}(e^{-tB(t)})(e^{-tB(0)}),
\end{align}
and $x_j(t):=e^{itH_Q}x_je^{-iH_Qt}$. Assuming a quadratic Hamiltonian for the quantum Brownian particle $H_Q:=\sum_i {\omega_i} q_i^\dagger q_i$, the time evolved operator $x_j(t)$ becomes simply $x_j(t)=\frac{1}{\sqrt{2}}(q_je^{-i\omega_jt}+q_j^\dagger e^{i\omega_j t})$. In this case we can combine all the products in $V_t$ \cref{eq:Vt} over a single site using the Baker-Campbell-Haussdorff (BCH) relation $e^Ae^B=e^{A+B}e^{\frac{1}{2}[A,B]}$, valid whenever $[A,B]$ commutes with $A$ and $B$. After some straighforward algebra we find
\begin{align}\nonumber
    V_t
    &=e^{-itH_Q}\prod_{k=1}^L \exp{\left(\frac{ig_k}{2}\sqrt{\frac{t}{2N}}\sum_{j=1}^N\left[{({q_ke^{-i\omega_k\frac{jt}{N}}+q_k^\dagger e^{i\omega_k\frac{jt}{N}}})Z_{j,k}}\right]\right)}e^{-\frac {ig_k^2t}{8N}\sum_{j>l}Z_{j,k}Z_{l,k}\sin\left(\omega_k\frac{(l-j)t}{N}\right)}.
    \label{eq:Vt_final}
\end{align}
Introducing the displacement operator $D(\alpha):=e^{\alpha^\dagger - \alpha}$, and defining  $\eta_k=\left[-\frac{ig_k}{2}\sqrt{\frac{t}{2N}}\sum_{j=1}^Ne^{-i\omega_k\frac{jt}{N}}Z_{j,k}\right]$ \cref{eq:rho_eph} becomes
\begin{align}
    \rho_F(t)=\sum_{\bm\gamma=\{0,1\}^{LN}}{\rm Tr}\left( e^{-itH_{\rm IQP}(t)} D(\bm q \cdot \bm\eta)(\rho_Q\otimes \Pi_{\bm 0}) D(-\bm q \cdot \bm\eta)e^{itH_{\rm IQP}(t)}\Pi_{\bm\gamma}\right){U}_{\bm\gamma,t}(\rho_F){U}^\dagger_{\bm\gamma,t}+O(t^2/N)
\end{align}
with $\bm q \cdot \bm\eta=\sum_{k=1}^Lq_k\eta_k$ and $H_{\rm IQP}(t):=\sum_k\frac {g_k^2}{8N}\sum_{j>l}Z_{j,k}Z_{l,k}\sin\left(\omega_k\frac{(l-j)t}{N}\right)$. Finally, assuming that the initial state of the quantum Brownian particle is the vaccum, and using the algebra of the displacement operator $D(\bm q \cdot \bm\eta)=e^{-\frac{1}{2} \bm\eta^*\cdot \bm\eta}e^{\bm q^\dagger \cdot \bm\eta^*}e^{\bm q \cdot \bm\eta}$ that follows from BCH, we have
\begin{align}
     {\rm Tr}_Q\left(D(\bm q \cdot \bm\eta)(\rho_Q\otimes \Pi_{\bm 0}) D(-\bm q \cdot \bm\eta)\right)=e^{-\frac{1}{2}\bm{\eta}^{*}\cdot\bm{\eta}}{\rm Tr}_{Q}\left(e^{\bm{q}^{\dagger}\cdot\bm{\eta}^{*}}|0\rangle\Pi_{\bm{0}}\langle0|e^{\bm{q}\cdot\bm{\eta}}\right)e^{-\frac{1}{2}\bm{\eta}^{*}\cdot\bm{\eta}}.
\end{align}
It is possible to compute the trace over the boson degrees of freedom by writing the projector $\Pi_{\bm 0}$ in the computational basis of the ancillary Ising qubits. In this basis the exponents on each side of the projector act as phases in the ancillary space. Concretely
\begin{align}
    {\rm Tr}_{Q}\left(e^{\bm{q}^{\dagger}\cdot\bm{\eta}^{*}}|0\rangle\Pi_{\bm{0}}\langle0|e^{\bm{q}\cdot\bm{\eta}}\right)=\frac{1}{2^{LN}}\sum_{\bm{s},\bm{s'}}|\bm{s}\rangle\langle\bm{s'}|\langle0|e^{\bm{q}\cdot\bm{\eta}(\bm{s'})}e^{\bm{q}^{\dagger}\cdot\bm{\eta}^{*}(\bm{s})}|0\rangle=\frac{1}{2^{LN}}\sum_{\bm{s},\bm{s'}}|\bm{s}\rangle\langle\bm{s'}|e^{\bm{\eta}^{*}(\bm{s})\cdot\bm{\eta}(\bm{s'})}
\end{align}
where we have used the BHC relation $e^{A}e^{B}=e^{[A,B]}e^{B}e^{A}$  and that the vacuum is invariant under the action of the exponential of bosonic destruction operators. Putting everything together we have
\begin{align}
    \rho_{F}(t)=\sum_{\bm{\gamma}=\{0,1\}^{LN}}{\rm Tr}\left(e^{t\mathcal{L}_{{\rm IQP}}(t)}(\Pi_{\bm{0}})\Pi_{\bm{\gamma}}\right)U_{\bm{\gamma},t}(\rho_{F})U_{\bm{\gamma},t}^{\dagger}+O(t^{2}/N)
    \label{eq:rho_f_t}
\end{align}
where
\begin{align}
    \mathcal{L}_{\rm IQP}(t)[\rho]:=-i[H_{\rm IQP}(t),\rho]+\sum_{k=1}^L \left(\ell_k(t)\rho\ell^\dagger_k(t)-\frac{1}{2}\left\{\ell^\dagger_k(t)\ell_k(t),\rho\right\}\right),
\end{align}
and $\ell_k(t)=\frac{g_{k}}{2}\sqrt{\frac{1}{2N}}\sum_{j=1}^{N}e^{i\omega_{k}\frac{jt}{N}}Z_{j,k}$.
\end{proof}
This is a remarkable result. It implies that the (dissipative) evolution of a fermion-boson system \cref{eq:Lind_eph} can be approximated to arbitrary precision with a fermion-qubit (dissipative) evolution. The evolution on the fermion-qubit side can be understood as the result of sampling the probability distribution of the ancillary qubits
\begin{align}
    \Gamma_{\bm{\gamma}}(t):={\rm Tr}\left(e^{t\mathcal{L}_{{\rm IQP}}(t)}(\Pi_{\bm{0}})\Pi_{\bm{\gamma}}\right),
\end{align}
and according to the outcome $\bm{\gamma}$ applying the unitary $U_{\bm \gamma,t}$ on the fermionic system. Classically sampling from this type of probability distributions, which fall under the category of \gls{iqp} circuits, even approximately, will result in  $\mathsf{BPP}^{\mathsf{NP}} =\#\mathsf{P}$ and cause a collapse in the polynomial hiearchy~\cite{Bremner_2016}. Finding an efficient classical protocol to do so would collapse the polynomial hierarchy at the third level \cite{Bremner_2010}. Of course, this does not prevent the classical simulation of this {\it particular} family of circuits, but nonetheless we think this connection is particularly interesting, hinting to the classical complexity of the problem and giving \gls{iqp} circuits a concrete application in quantum simulation.

In the next section we discuss a circuit construction to approximate the Linbladian on the ancilla qubits, thus providing a concrete realisation of \cref{lem:boson_qubit}.

\subsection{Circuit construction of auxiliary dissipative Lindbladian}
\label{sec:iqp_details}

Using \cref{lem:boson_qubit}, we can approximate the density matrix of the fermionic system by sampling the bitstring ${\bm \gamma}$ from the probability $\Gamma_{\bm \gamma}(t)$. This corresponds to sampling from an \gls{iqp} circuit, where all the gates are commuting. The generator of this evolution is the Linbladian \cref{eq:iqp} 
composed of two contributions, one given by $H_{\rm IQP}$ which can be trivially implemented using a circuit with no ancillas, and the dissipative part whose circuit approximation we discuss here. Writing explicitly the action of these contributions we find
\begin{align}
    \Gamma_{\bm{\gamma}}(t)&= \langle\oplus|U_{\rm IQP}(t) e^{t\mathcal{D}_{\rm IQP}(t)}(\Pi_{\bm 0}) U^\dagger_{\rm IQP}(t)   |\oplus\rangle
\label{eqn:Gamma_explicit}
\end{align}
where $\ket{\oplus} = (H\ket{0})^{\otimes NL}$ is the superposition of all computational basis states on $NL$ qubits and $H$ is the Hadamard gate, while
\begin{equation}
U_{\rm IQP}(t) =e^{-itH_{\rm IQP}(t)}
\end{equation} 
is the unitary part of the evolution. Note that for the \gls{iqp} circuit time is a parameter that determines the circuit and not a usual time variable of a time independent Hamiltonian. The purely dissipative part is generated by 
\begin{align}
\mathcal{D}_{\rm IQP}(t)[\rho]:=\sum_{k=1}^L \left(\ell_k(t)\rho\ell^\dagger_k(t)-\frac{1}{2}\left\{\ell^\dagger_k(t)\ell_k(t),\rho\right\}\right)=\sum_k \mathcal{D}_k(t)[\rho],
\label{eqn:DIQP}
\end{align}
with the jump operators $\ell_k(t)=\frac{g_{k}}{2}\sqrt{\frac{1}{2N}}\sum_{j=1}^{N}e^{i\omega_{k}\frac{jt}{N}}Z_{j,k}$.

We can approximate the density matrix $\rho_k(t):=e^{t\mathcal{D}_{k}(t)}[\rho]$ with arbitrary small error by defining the map 
\begin{equation}
    G_{k,\epsilon}[\rho]:={\rm Tr_{a}}(e^{\frac{i}{2}\epsilon((\ell_k+\ell^{*}_k)X-i(\ell_k-\ell_k^{*})Y)}\rho|0\rangle\langle0|_{a}e^{-\frac{i}{2}\epsilon((\ell_k+\ell_k^{*})X-i(\ell_k-\ell_k^{*})Y)})
    \label{eqn:Gkepsilonmap}
\end{equation}
with an ancillary qubit $a$ initialised in the $\ket{0}$ state. Here the Pauli matrices $X,Y$ act on ancillary space of the qubit $a$. This map
approximates $\mathcal{D}_{k}$ with an error of order $\epsilon^4$ as by taking the trace over the ancillary qubit we find
\begin{align}
G_{k,\epsilon}[\rho]&=\cos(\epsilon|\ell_k|)\rho\cos(\epsilon|\ell_k|)+\sin(\epsilon|\ell_k|)e^{i\ell_k/|\ell_k|}  \rho\sin(\epsilon|\ell_k|)e^{-i\ell_k^*/|\ell_k|},\\
&=\rho+\epsilon^{2}\left(\ell_k(t)\rho\ell_k^{\dagger}(t)-\frac{1}{2}\{\ell_k^{\dagger}(t)\ell_k(t),\rho\}\right)+O(\epsilon^{4}) 
=e^{\epsilon^{2}\mathcal{D}_k}(\rho)+O(\epsilon^{4}),
\label{eqn:gmap_error}
\end{align}
and the Trotterised evolution of $G$ into $N$ slices (with conventional error scaling as $\epsilon^2 = \frac{t}{N}$) gives
\begin{equation}
    e^{t\mathcal{D}_k(t)}[\rho]=\left(e^{\frac{t}{N}\mathcal{D}_k(t)}\right)^{N}[\rho]=\left(G_{k,\sqrt{\frac{t}{N}}}[\rho]\right)^N+O\left(\frac{t^2}{N}\right).
\end{equation}

Given the commuting property of $\mathcal{D}_{\rm IQP}$ it is possible to use either multiple ancillas to perform the $N$ applications of $G_{k,\sqrt{\frac{t}{N}}}$ and trace, or use one ancilla with reset, without the need for feedforward. This has to be repeated $L$ times as $e^{t\mathcal{D}_{\rm IQP}(t)}=\prod_{k=1}^Le^{t\mathcal{D}_{k}(t)}$. 
In practice if the product of exponentials of Paulis are the implementable gates available, implementing the map $G_{k,\epsilon}$ in \cref{eqn:Gkepsilonmap} requires a further Trotterisation. Using a $p$-order product formula $S_p(t)$ that approximates $G_{k,\sqrt{\frac{t}{N}}}$ we have
\begin{align}\label{eqn:ancilla_lindblad_map}
    G_{k,\sqrt{\frac{t}{N}}}=(G_{k,\frac{1}{N_{\rho}}\sqrt{\frac{t}{N}}})^{N_{\rho}}=S_{p}\left(\frac{1}{N_{\rho}}\sqrt{\frac{t}{N}}\right)+O\left(\frac{1}{N_{\rho}^{p}}\left(\frac{t}{N}\right)^{\frac{p+1}{2}}\right).
\end{align}
To retain the same overall error of the Trotterisation in \cref{eq:rho_F_lem}, we can choose for example $p=1$, $N_\rho=O(N)$, $p=2$ and $N_\rho=O(N^{1/4})$ or $p=4$ and $N_\rho=O(1)$. 
\subsection{Resource estimates for electron-phonon coupled unital Lindbladians}
This section discusses the quantum resource estimates to simulate open quantum systems pertaining to electron-phonon interactions as described in \cref{eq:Lind_eph}. The subsystems consist of arbitrary fermions and harmonic bosons, and they interact locally on each site via the bosonic displacement operator, $x$ and fermionic number operator, $n$. \cref{table:resource_estimate} shows the resource estimates of the Trotterised Lindbladian dynamics governed by \cref{eq:Lind_eph} in a noiseless simulation. This cost is for obtaining the unbiased expectation of some observable $P^A_{\text{obs}}$ in subsystem $\mathcal{S}_A$ with error $\epsilon$ up to some simulation time $T$ and $R$ Trotter steps. We neglect the cost of initial state preparation and the final rotation onto the eigen-basis of $P^A_{\text{obs}}$. For the cost of implementing boson operators on a qubit processor, the boson Hilbert space has to be truncated to $N_b$ levels per site and we quote the asymptotic scaling for encoding the tridiagonal matrix operator $x$ in the Binary/Gray encoding of $O(N_b^2 log_2 N_b)$~\cite{sawaya2020resource} and will discuss this choice in more detail later in the section. It is worth noting that by using the stochastic averaging (\cref{thm:1}) method (or any method that requires encoding the boson Hilbert space), we inherently assume that $N_b$ is chosen such that consequent truncation errors up to the maximum simulation time do not exceed $\epsilon$. From here on we refer to the stochastic averaging method as SA. 

Note that $Q(H_F)$ may also have a dependence on the system size, determined by for example the connectivity of the fermionic Hamiltonian. For the worst case $2q$ gate qubits cost of Dcube in \cref{table:resource_estimate}, we employed a $p=4$ order product formula with a single ancillas for realising the approximation of \cref{eqn:ancilla_lindblad_map}. A smaller choice of $p$ will result in the same asymptotic scaling due to the $R^2$ cost of implementing $H_{IQP}$ as defined in \cref{eq:iqp}. Recall that the value of $p$ dictates how many applications of the map $G$ (\cref{eqn:Gkepsilonmap}) is required. Since each map application can be either done by a single ancilla via measure and reset or multiple ancillas, the optimal choice of $p$ and ancilla usage will be hardware dependent. 

\begin{table}[H]
\begin{center}
\begin{tabular}{|c | c | c | } 
 \hline
 & Stochastic Averaging (SA) (\cref{thm:1})  & Dcube (\cref{theo:Dcube})\\
 \hline
$\#$ Circuits & $O(\epsilon^{-2})$ &  $O(\epsilon^{-2})$   \\
\hline
 $\#$ WC 2q gates & $R \cdot Q(H_F) + O\left( R\cdot (N_L N_b^2 \text{log}_2 N_{b})\right)$ &  Max($O(R^2)$, $O(R\cdot Q(H_F))$)\\
\hline
$\#$ WC Qubits & $N_L \cdot \left( \lceil{\text{log}_2 N_b \rceil} + 2\right)$ & Max($R+1,2N_L$)  \\
\hline 
$\#$  WC total $N_{\text{samples}}$ & $O(\epsilon^{-2})$ & O$(\epsilon^{-2})$\\
\hline
\end{tabular}
\caption{Resource costings for simulating the Trotterised electron-phonon Lindblad evolution specified in \cref{eq:Lind_eph}, comparing stochastic averaging (\cref{thm:1}) vs Dcube (\cref{theo:Dcube}). We assume a fixed total simulation error $\epsilon$, the number of sites $N_s$, boson truncation level $N_{b}$ per site and number of Trotter steps $R$. The row names including `WC' denote the circuits in the algorithm with the \textit{worst case} scaling in qubts, arbitrary 2q-gates and total number of samples. We denote $Q(\cdot)$ as the number of 2-q gates required for the qubit encoding of the desired operator. Note that in the 2q gate count, the stochastic averaging method requires the encoding of the bosonic displacement operators $\hat{x}$, which we assume quote to scale as $O(N_b^2 log_2 N_b)$~\cite{sawaya2020resource}.}
\label{table:resource_estimate}
\end{center}
\end{table}


There are more recent literature studies on reducing the cost of bosonic operator to qubit encodings. Particularly in \cite{liu2026hybrid}, specific numbers for two qubit gate counts are provided for the bosonic displacement operator $\hat{x}$ as a function of boson truncation level $N_b$. Their method allows for (in theory) a scaling of $O(\text{log}_2 N_b)^2$ but introduces additional ancilla qubits and includes extra Trotter error. Additionally, their cost also scales according to $m$, the number of iterative Newton steps required for implementing the square root operation on a quantum computer. Furthermore, in the fixed range they considered of $N_b=1-1000$, this method does not beat the best brute force $n$-qubit decomposition~\cite{krol2024beyond}. However, it offers an algorithmic method of determining the gate decomposition for $\hat{x}$ when the brute force decomposition becomes classically intractable. 

Thus we will use the brute force CNOT cost as concrete numbers to compare against. For simplicity we ignore the fermionic cost in two qubit gate count and assume it shifts all gate counts by a constant factor. We also fix a desired error of $\epsilon \sim 10^{-2}$ and a total time evolution to $T=\sqrt{10}$ such that the number of first order Trotter steps needed is roughly $\sim 1000$.  The qubit and CNOT count for simulating  \cref{eq:Lind_eph} for various system sizes $N_s$ are show in \cref{fig:resource_cost_scaling}. The CNOT and qubit costings for the $SA$ method only include that of implementing the bosonic $\hat{x}$ operator. Using the brute force decomposition method~\cite{krol2024beyond}, the total CNOT cost required for $R$ Trotter steps is $\lceil R \times N_s \times (\frac{22}{48} N_b^2 - \frac{3}{2}N_b  + \frac{5}{3})\rceil $. On the other hand, the qubit count is $N_s \times (\lceil \text{log}_2 N_b \rceil +2)$, where the 2 results from assuming spinful fermions.

\begin{figure}[H]
\centering
\includegraphics[width=\textwidth]{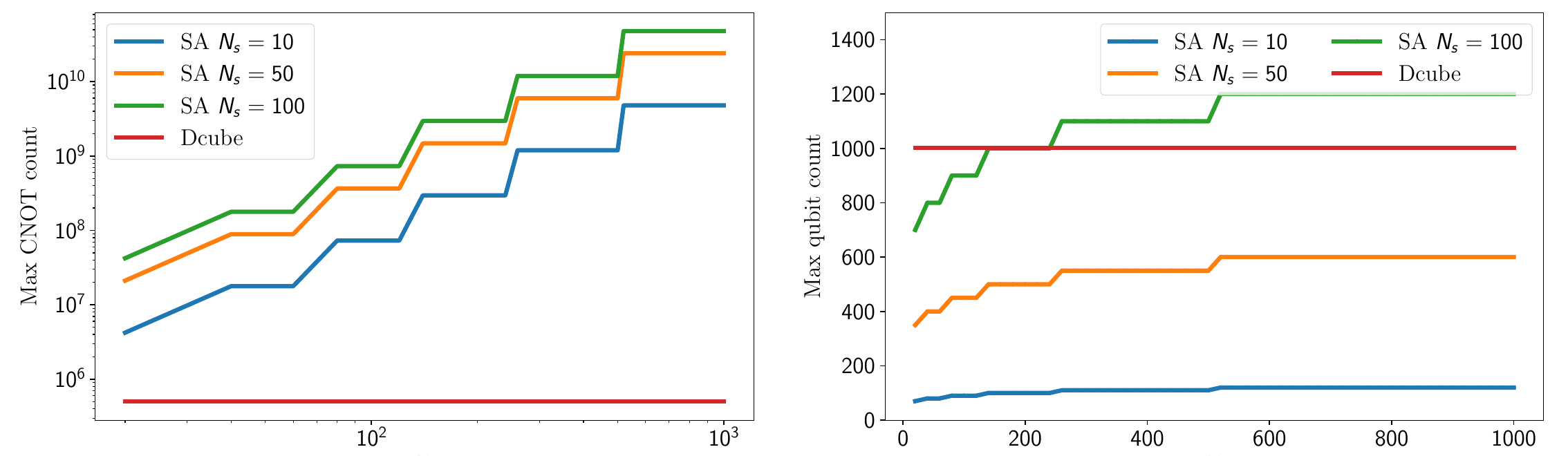}
\caption{(Left) CNOT gate and (Right) qubit count for simulating the Lindbladian dynamics, \cref{eq:Lind_eph}. Here we compare between Dcube (red lines) against SA, which requires boson operator to qubit encoding. For $R=1000$ Trotter steps, and various system sizes $N_s$, we plot the resources required in SA for encoding the bosonic displacement operator $\hat{x}$ in the energy basis explicitly. This operator exists in the Hamiltonian across all sites, and gets more costly with increasing Hilbert space truncation $N_b$. This cost is compared against the resources for the \textbf{full} dissipative Lindbladian simulation with Dcube. The brute force method of arbitrary gate decomposition discussed in ~\cite{krol2024beyond} is used for the two qubit costing of the $\hat{x}$ operator.}
\label{fig:resource_cost_scaling}
\end{figure}

\Cref{fig:resource_cost_scaling} demonstrates the simple message that boson encoding is expensive. The Dcube algorithm completely decouples the fermion and boson subsystems without requiring coherent measurement and feedforward, and the reported `max' cost is simply the more expensive one of the two subsystems. Furthermore, we emphasise that the costs for SA (or any other method that requires boson to qubit encoding) is just that of the $\hat{x}$-operator across all the $N_s$ sites (each with boson Hilbert space truncation $N_b$), not the full dissipative Lindbladian simulation. On the other hand, the quantum resources reported for Dcube (red lines) is enough for the entire Lindbladian \gls{tds}. \Cref{fig:resource_cost_scaling} also clearly displays the system size and $N_b$ independence of Dcube.


For the general Lindbladian simulation, more sophisticated methods that beat the $O(t^2/\epsilon)$ scaling exists~\cite{cleve2016efficient, li2022simulating,peng2025quantum}, but often require e.g. block-encodings, linear combination of unitaries and amplitude amplification, which Dcube does not neccessitate. 
For the \textit{unitary} evolution part of the \gls{tds}, we have quoted errors scaling in terms of the deterministic Trotter product formulas. Alternatively,  methods such as Quantum Signal Processing~\cite{low2017optimal, low2019hamiltonian} asymptotically achieve optimal linear scaling in time and polylogarithmic in $1/\epsilon$, but we focus here on those that are more compatible for near term hardware implementation. More sophisticated Trotter only methods also exist that can achieve better than $O(t^2)$ scaling. For example, Ref. \cite{rendon2023all} proposes an algorithm for amplitude simulation using only Trotter and classical postprocessing with resources scaling linearly in time. This can be useful for example in reducing the cost of simulating the observable in $\mathcal{S}_A$ when \gls{iqp} samples corresponding to the effects of $\mathcal{S}_B$ have been generated. Another Trotter only based method to note is that of randomisation~\cite{childs2019faster}, whereby randomising the Trotter layer ordering allows a first order implementation to scale as second order. This reordering approach applied on higher order formulae also result in a reduction of simulation error, but not in an increase in effective order. However, due to the overall scaling of $O(t^2)$ introduced by approximating the Lindbladian evolution with \cref{thm:1}, techniques such as randomisation will not help the asymptotic scaling. However, they can be used to improve the unitary components such as that of the fermionic Hamiltonian \gls{tds} (i.e. $Q(H_F)$)  and \cref{eqn:Gkepsilonmap}.  We believe the analysis presented herein indicates that Dcube is better suited to current hardware constraints compared to methods that require boson operator to qubit encoding. 

\subsection{Nontrivial steady states of $\mathcal{L}_{\rm IQP}$}

The ancillary Lindbladian \cref{eq:iqp} has one trivial fix point corresponding to the fully mixed density matrix $ \mathbb{I}/2^{LN}$ as $\mathcal{L}_{\rm IQP}(t)[\mathbb{I}]=0$. This motivates the question: Does the dynamics in the ancillary system tend to a trivial infinite temperature state? As the Linbladian is diagonal in the computational basis, we can read off directly its eigenvalues. The real part of the eigenvalues controls the decay, so we concentrate on the eigenvalues of $\mathcal{D}_{\rm IQP}(t)$. The eigenmatrices of the dissipator $\mathcal{D}_{\rm IQP}(t)$ are $\sigma_{\bm s,\bm    s'}:=|\bm{s}\rangle\langle \bm{s}'|$ where $\bm{s}$ is a bitstring in the computational basis of the ancillas. The eigenvalues $\mathcal{D}_{\rm IQP}(t)\sigma_{\bm s,\bm s'}=\lambda_{\bm s,\bm s'}\sigma_{\bm s,\bm s'}$ are trivially
\begin{align}
    \lambda_{\bm s,\bm s'}&=\sum_k\left[\ell_{k,\bm{s}}(t)\ell^*_{k,\bm{s'}}(t) -\frac{1}{2}\left(|\ell_{k,\bm{s}}(t)|^2+|\ell_{k,\bm{s}'}(t)|^2\right)\right]=-\frac{1}{2}\sum_k\left|\ell_{k,\bm{s}}(t)-\ell_{k,\bm{s}'}(t)\right|^{2}+i\sum_k\Im\left(\ell_{k,\bm{s}}(t)\ell_{k,\bm{s}'}^{*}(t)\right),
\end{align}
with $\ell_{k,\bm s}(t)=\frac{g_{k}}{2}\sqrt{\frac{1}{2N}}\sum_{j=1}^{N}e^{i\omega_{k}\frac{jt}{N}}{\bm s}_{j,k}$. The real part of the eigenvalue controls the decay of the corresponding eigenmatrix. Note that for the diagonal eigenmatrices $\sigma_{\bm s,\bm s}$ the real part vanishes. In order to understand the approach to a steady state, we need to see how the Lindbladian gap behaves as the number of Trotter steps increases to keep a fixed error for larger times. We can upper bound the real part of the gap by looking at the eigenmatrices $\sigma_{\bm s, \bm {s}+ \bm{\eta}}$ where the sum in the index is understood modulo 2. For these eigenmatrices, the real part of the eigenvalue is
\begin{align}
    |\Re(\lambda_{\bm s, \bm s + \bm \eta})|=\frac{1}{2}\sum_k\frac{g_k}{8N}\left|\sum_{j=1}^Ne^{i\omega_{k}\frac{jt}{N}}{\bm \eta}_{j,k}\right|^2\leq \frac{g_{\rm max}L}{16N}|\bm \eta|^2
\end{align}
where  $g_{\rm max}=\max_k g_k$, and $|\bm \eta|$ is the Hamming weight of the bitstring $\bm \eta$. This implies that for any off diagonal element $\sigma_{\bm s, \bm s + \bm \eta}$ with $|\bm \eta|=o(\sqrt{N})$ the real part of $\lambda_{\bm s, \bm s +\bm \eta}$ vanishes in the limit $N\rightarrow\infty$. Then the long time limit of the evolution is nontrivial as the states (in the computational basis) that survive the decay live on a band of size $o(\sqrt{N})$ around the diagonal.


\section{Numerical results}
\label{sec:numerics}

\subsection{Free fermion dimer coupled to bosons} \label{sec:dimer}
So far we have discussed theoretically a framework for approximating dephasing Lindbladian evolution. In \cref{sec:elec_phonon} we provide explicit error bounds and a precise protocol for the specified type of electron-phonon interaction. In this section, we numerically demonstrate the results presented in \cref{lem:boson_qubit}, which studies how to simulate the electron-phonon Lindbladian described by \cref{eq:Lind_eph}. Here we perform detailed analyses on a spinless fermionic system of two sites shown in \cref{fig:schematic_dimer_dephasing_lindbladian_model_system}, also known as the dimer. Based on this we will refer to this system as the dimer in the following discussion. 

\begin{figure}[H]
\centering
\includegraphics[width=0.6\textwidth]{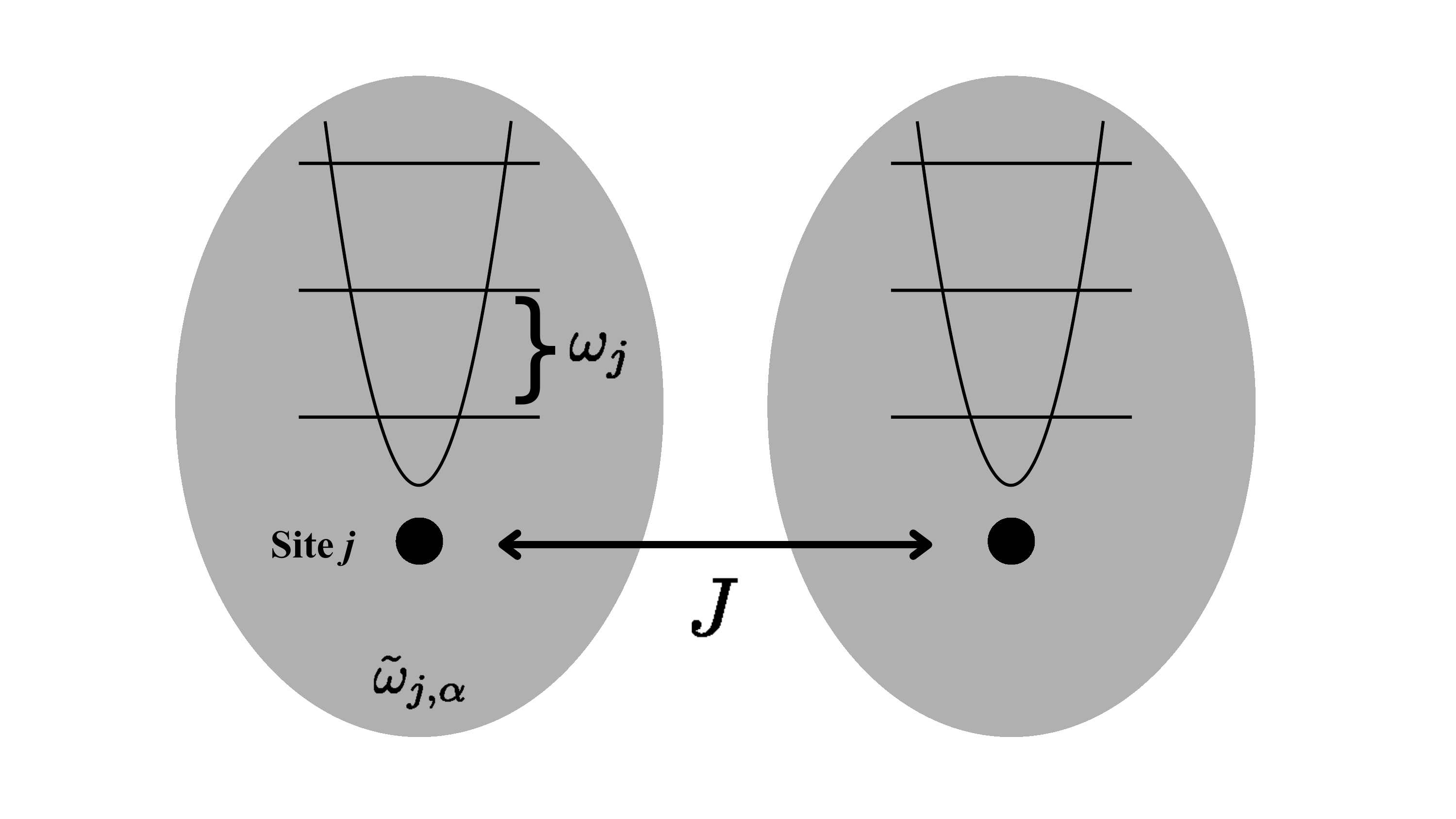}
\caption{Schematic of the dimer in the presence of a bath and quantum Brownian particles on each site. Physical sites are denoted by the black circles with site index $j$. The harmonic oscillator with energy level spacing $\omega_j$ denote the Brownian particle on each site, relating to the Hamiltonian $H_Q$. The gray oval encompassing the entire site represents the boson bath with a continuum of energy levels $\tilde{\omega}_{j,\alpha}$ (see also \cref{appendix:epi_lindbladian_derivation}). On each site, the bath, Brownian particle and fermions are all coupled together with interaction strength $g_j$. Fermions can hop between the sites with hopping strength $J$ and are described by the Hamiltonian $H_F$.}
\label{fig:schematic_dimer_dephasing_lindbladian_model_system}
\end{figure}

More specifically, the fermion and boson subsystem Hamiltonians from \cref{eq:Lind_eph} under consideration are 
\begin{align}\label{eqn:HF_HQ}
    \centering
    H_F &= \sum_{\langle i,j \rangle, \sigma } -J(c^{\dagger}_{i,\sigma} c_{j,\sigma} + c^{\dagger}_{j,\sigma} c_{i,\sigma}) + U n_{i,\uparrow} n_{i,\downarrow}, \nonumber \\
    H_Q &= \sum_j \omega_j q^{\dagger}_j q_j
\end{align}
where $J$ is the fermion hopping amplitude across nearest neighbours, $n_{i,\sigma}$ the fermion number operator for site $i$ and spin $\sigma=\uparrow$ (WLOG we restrict to the spin up sector), and  $\omega_j$ the harmonic oscillator frequency for the Brownian particle on site $j$. For this section, we set $J=1$, $U=0$, $\omega_j = 1$, $g_j = g, \forall j$ and consider the system with one fermion. Wherever it is relevant, we truncate the boson Hilbert space to contain at most eight energy levels ($N_b=8$). In general for a system on $L$ sites with $N_b$ boson levels per site, the number of qubits needed for the simulation scales as $L(\log_2(N_b)+1)$. This classical simulation scales very unfavorably since the size of the matrices involved at worst grows as $2^{2L(1+\text{log}_2(N_b))}=2^{16}$ for this spinless case, where the additional factor of two in the exponential comes from the full density matrix simulation. 

We choose to initialise the state of the Brownian particles in the vacuum and populate a single fermion on site one. The fermion density dynamics for this site, $\text{Tr}(n_1 e^{t\mathcal{L}} \rho(0))$ (where $\mathcal{L}$ is given by \cref{eqn:L_epi}) is simulated with various methods and shown in \cref{fig:TDS_1x2_gpaper_4.0_10000_samples_and_density_errors_ALLMETHODS}. The solid lines named `exact evolution' are the density matrix time evolutions according to \cref{eq:Lind_eph}. We will refer to this method as `exact Lindbladian'. 
It is important to note that the `exact Lindbladian' evolution also suffers from error, primarily from the truncation of the infinite dimensional Hilbert space associated with the Brownian particles. We expand upon this problem in \cref{sec:boson_scaling_problem}. A na$\ddot{\i}$ve strategy to tackle this is to study convergence of observables by incrementally changing the boson Hilbert space cutoff. However this can be computationally demanding, highlighting another issue with general simulations involving bosons. Therefore, when comparing Dcube against `exact Lindbladian' for general simulations involving bosons, it is imperative to account for the effect of boson truncation on the fermion observables, because Dcube implicitly takes into the account the infinite bosonic Hilbert space. 

\Cref{thm:1} allows us to directly simulate the `exact Lindbladian' evolution, through the expectation over stochastic channels. This is denoted as the `stochastic' method and cuts down the statevector simulation cost down to $2^{L(1+\text{log}_2(N_b))}$. The `circuit' curves correspond to the Dcube algorithm, with $N$ Trotter steps to approximate the true Lindbladian evolution (without truncating the boson Hilbert space) by sampling from the probability distribution of the ancillary qubits, $\Gamma_{\bm{\gamma}}(t)$ given by \cref{eqn:Gamma_explicit}. The fermion density is then computed via a Monte Carlo estimate of the full summation 
\begin{equation}
    \centering
    \text{Tr}(n_i \rho_F(t))_{\text{circuit}} = \frac{1}{N_s} \sum_x \text{tr}\left( n_i U_{\bm{\gamma}_x, t} (\rho_F) U^{\dagger}_{\bm{\gamma}_x,t} \right) 
\label{eqn:montecarlo_exactcircuit}
\end{equation}

with $N_s$ samples. The errors associated with the `stochastic' and 'circuit' methods are computed as $\frac{\sigma}{\sqrt{N_{\text{s}}}}$ where $\sigma$ denote the standard deviation od a particular observable computed from $N_{\text{s}}$ samples. Since in our case the fermion subsystem $H_F$ is non-interacting, \cref{eqn:montecarlo_exactcircuit} is computed efficiently classically via \gls{flo}~\cite{terhal2002classical} after the circuit  samples are collected. In the implementation of $\Gamma_{\bm{\gamma}}(t)$ in \cref{eqn:Gamma_explicit}, a second order Trotter decomposition is used to realise the map for the dissipative evolution $e^{t\mathcal{D}_{IQP}(t)}$ in \cref{eqn:DIQP}. In order to retain an overall error of $O\left( \frac{1}{N} \right)$, the number of these internal Trotter steps is set to $N_{\rho} = \lceil {N^{\frac{1}{4}}} \rceil$ and $N_{\eta} = N$ ancillas are used to realise \cref{eqn:DIQP}.

\begin{figure}[H]
\centering
\includegraphics[width=\textwidth]{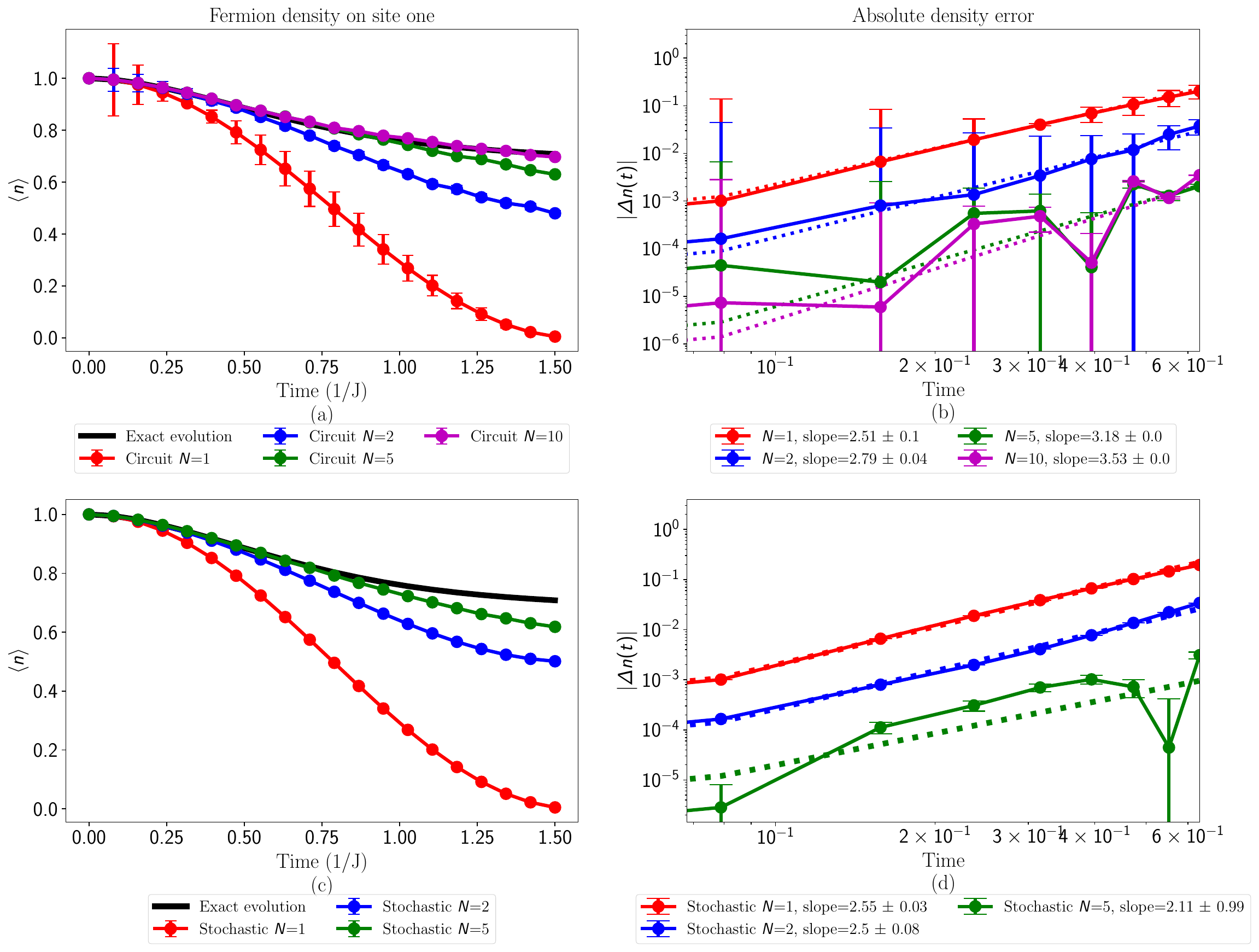}
\caption{Density time dynamics of the spinless free fermion dimer induced by the dephasing Lindbladian of \cref{eq:Lind_eph} compares (a) `exact Lindbladian' against `circuit' method with the fermion density absolute error in (b). The dashed lines are linear fits to the logarithm of the absolute density error with the slopes in the respective legend according to $N$ Trotter steps. The number of ancillas to realise the auxiliary dissipative Lindbladian in \cref{eqn:DIQP} is always set equal to the total number of Trotter steps, $N_{\eta}=N$. Each application of this map uses $N_{\rho}=\lceil {N}^{\frac{1}{4}}\rceil$ internal Trotter steps. A total of $10$k samples are taken per experiment across all different Trotter steps for the `circuit' method. Similarly, the time dynamics of (c) `exact Lindbladian' against `stochastic' method with absolute fermion density error are compared in (d). $5$k samples of uniformly sampled stochastic channels are used across all $N$ Trotter steps. $N_b=8$ is set for `exact Lindbladian' and `stochastic' simulations. The initial state is a Fock state with a single fermion on site one and $g=4$ across all simulations.}
\label{fig:TDS_1x2_gpaper_4.0_10000_samples_and_density_errors_ALLMETHODS}
\end{figure}

\Cref{fig:TDS_1x2_gpaper_4.0_10000_samples_and_density_errors_ALLMETHODS} also exhibits fermion number density error scaling  $|\Delta n^f|$  between the `exact Lindbladian' vs Trotter approximations on the right column. Visibly, the logarithmic error grows linearly in time with a slope larger than the one expected from the overall first order Trotter scaling. For reference, the expected sampling noise, $\epsilon_{\text{sampling}}$, at $N_s=10$k is on the order of $\epsilon_{\text{sampling}} \sim 10^{-2.4}$ which contributes to the standard deviation of the results. Furthermore, the boson truncation in the `exact Lindbladian' simulations contribute to errors on a similar order of magnitude, as examined below in \cref{sec:boson_scaling_problem}. 

\subsubsection{Dependence on Hilbert space truncation}\label{sec:boson_scaling_problem}

Consider again time dynamics using the `exact Lindbladian' simulation method where the boson Hilbert space is truncated. The boson number density evolution is shown on the LHS of \cref{fig:Lindblad_exact_nxbosonscaling_1x2_gpaper_4.0}. Assume there is no knowledge of the $N_b=4$ curves. Notice that the boson number density for $N_b=2$ on each site, shown as dashed lines in \cref{fig:Lindblad_exact_nxbosonscaling_1x2_gpaper_4.0}, never reaches the maximum allowed value of one, but appears to equilibrate well below it. Therefore, one may be inclined to set the cut-off at $N_b=2$. However, increasing the truncation level result in the bosons following completely different dynamics and settling into a new, higher equilibrium. This reflects the fact that truncating the boson subspace \textit{changes} the Hamiltonian for which the dynamics are simulated. This effect also manifests in the fermion dynamics where there is visible deviation between the $N_b=2$ vs $N_b=4$ fermion densities. The RHS of \cref{fig:Lindblad_exact_nxbosonscaling_1x2_gpaper_4.0} shows the absolute difference in fermion density using different truncations of the bosonic Hilbert space. This variation is on the scale of $ 10^{-2}$ when comparing $N_b=4$ and $N_b=8$. 
In light of this, one should be mindful when comparing Dcube against `exact' simulations such as that in \cref{fig:TDS_1x2_gpaper_4.0_10000_samples_and_density_errors_ALLMETHODS}, because the latter always involves a truncation of the bosonic Hilbert space.


\begin{figure}[H]
\centering
\includegraphics[width=\textwidth]{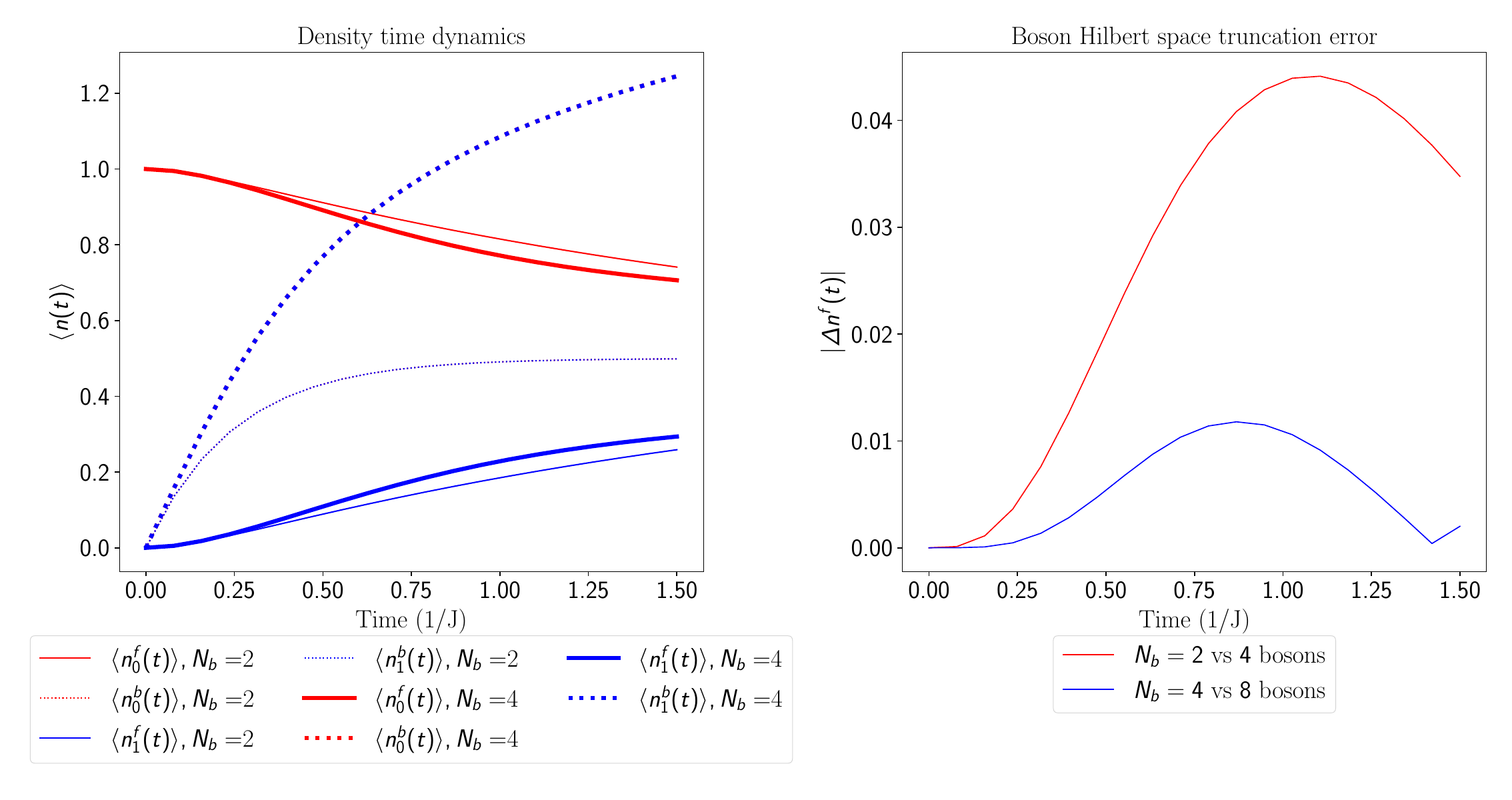}
\caption{(LHS) Time dynamics of particle density of a spinless Fermi-Hubbard dimer under the Lindbladian \cref{eq:Lind_eph}. $\langle n^d_{i} \rangle$ represents the number density of site $i \in [0,1]$ and particle identity $d \in [f,b]$ (f=fermion, b=boson), for two instances of $N_b$. The initial state is a single fermion on site $i=0$. Each plot is associated with the boson truncation level of maximum $N_b-1$ bosons per site. (RHS) The maximum change in fermion density across all sites with increasing boson Hilbert space dimension in the `exact Lindbladian' simulations for the dimer at $g=4$.}
\label{fig:Lindblad_exact_nxbosonscaling_1x2_gpaper_4.0}
\end{figure}

The effect of the truncation is dependent on the interaction strength.  To see this, a particular limit is illuminating. Consider for example the Holstein Hamiltonian~\cite{holstein1959studies} in real space using the same conventions for fermions and bosons as we have throughout the rest of the text,
\begin{equation}
    H = H_F + \frac{1}{2}  \sum_i\left(  \omega_i^2 x_i^2 + p_i^2 \right)+  \sum_i g_i x_i c^{\dagger}_i c_i,
\label{eqn:full_hamiltonian_before_trace_realspace}
\end{equation}
where $p_i$ is the momentum operator of the boson particle. When the fermionic occupation is full, such that $\langle n_{i,f} \rangle$ is constant and fixed for all sites $i$, the purely bosonic Hamiltonian can be written as
\begin{equation}\label{eq:holstein_shifted_boson_energy}
    H - H_F = \frac{1}{2}\sum_i \left(\omega_i x_i + \frac{g_i \langle n_{i,f} \rangle}{\omega_i}\right)^2 + \frac{1}{2}p_i^2
\end{equation}
and is valid up to adding a constant. In this regime, it follows from \cref{eq:holstein_shifted_boson_energy} that the boson position operator shifts to a new equilibrium depending on $g_i$. Intuitively, a higher average displacement means a higher harmonic oscillator energy and hence a larger boson population. Consequently this simple limit shows that truncating the bosonic Hilbert space has to be done appropriately, carefully studying the convergence of the results as a function of truncation. 
For a similar analysis on the behavior of a different observable with changing the boson truncation level, see \cref{appendix:gf_dcube_vs_exact_tds} following the discussion in \cref{sec:gf_numerics_1d}.

\subsubsection{Behaviour of auxiliary dissipative Lindbladian $\Gamma$ distribution} 

In this work we have uncovered a relationship between dissipative Lindbladian dynamics and \gls{iqp} circuits. This suggests that although the evolution is dissipative in nature, it may be hard to capture it classically. This is because to carry out this computation it is necessary to sample from \gls{iqp} circuits. In this section, we show numerically that the probability distribution arising from these circuits for small system sizes quickly become non-trivial, which we take as an indication that sufficiently long time simulations can be difficult. Theoretically, it has been shown that, assuming the polynomial hierarchy does not collapse, efficiently simulating the natural outputs of \gls{iqp} can be classically hard~\cite{bremner2015average,bremner2011classical,aaronson2011computational}. In \cref{fig:1x2_gpaper_4.0_Ntrotter_2_time_1}, we explicitly present the behaviour of the probability distribution \cref{eqn:Gamma_explicit} for the spinless dimer example discussed in \cref{sec:dimer}, for various times $T$ with two first order Trotter step (i.e. a total of four qubits). It shows that an initially peaked distribution at $T=0$ becomes anti-concentrated as time increases.
\begin{figure}[H]
\centering
\includegraphics[width=\textwidth]{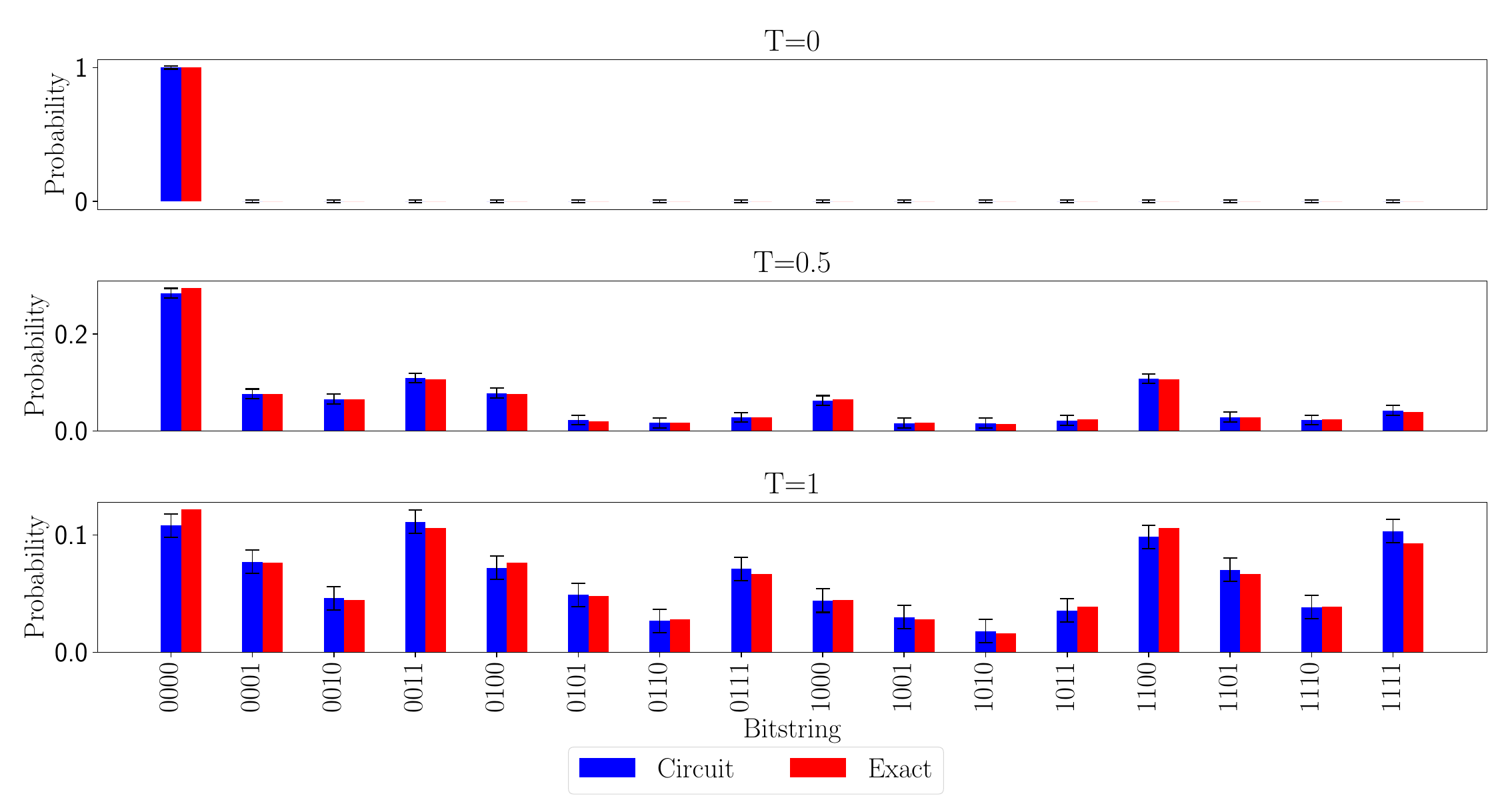}
\caption{The full decomposition of the probability distribution of the dimer $\Gamma(t)$ for different bitstring $\bm{\gamma}$s at various fixed total evolution time. $g=4$, $\omega=J=1$, $N_{\eta}=5$ and $N_{\rho}=2$ and $N=2$ first order Trotter steps. The `exact' method refer to exactly computing $\Gamma(t)$. The `circuit' method refers to sampling from the corresponding noiseless \gls{iqp} circuit that approximates $\Gamma(t)$. $N_s=10$k samples were drawn and the error bars signify the sampling error $\frac{1}{\sqrt{N_s}}$.}

\label{fig:1x2_gpaper_4.0_Ntrotter_2_time_1}
\end{figure}
The `exact' implementation is the exact computation of \cref{eqn:Gamma_explicit}. The `circuit' results are computed by sampling from the noiseless circuit that approximates $\Gamma$ as discussed in \cref{sec:iqp_details}. Note that due to site independence, rather than sampling once from a circuit of $LN$ qubits, one can also stitch together $L$ samples of a circuit of $N$ qubits. At $t=0$, \cref{eqn:Gamma_explicit} is simply the all-zero bitstring and the $\Gamma(0)$ is simply a delta function, melting into a different \gls{iqp} distribution at each time.

\subsection{Green's function of the spinful Fermi-Hubbard dimer under various dephasing}
\label{sec:gf_numerics_1d}
In this section, we show the usefulness of \cref{thm:1} and \cref{theo:Dcube} in cutting down the computational cost compared to exact Lindbladian simulation. We consider the effects of two different dephasing Lindbladians on the interacting Fermi-Hubbard dimer at half-filling with zero spin along the z-direction, i.e. $S_z=0$ (such that there is a total of one spin up and one spin down electron). Consider such a system subject to two dephasing Lindbladians. The first is the electron-phonon Lindbladian described by \cref{eq:Lind_eph}. The second is similar but purely fermionic in nature, given by \cref{eqn:purely_fermion_lindbladian}
\begin{equation}
    \centering
    \mathcal{L}_e(\rho)  := -i[H_F,\rho] + \sum_{j=1}^L g_j^2 \mathcal{L}_{n_j}(\rho).
    \label{eqn:purely_fermion_lindbladian}
\end{equation}
This type of Lindbladians describes the usual dephasing of electrons in materials and has been recently studied in refs~\cite{picano2025heating,alba2025free} in the context of dissipation induced heating dynamics. In this section, we consider the spinful dimer so the spin index in \cref{eqn:HF_HQ} runs over $\sigma \in [\uparrow, \downarrow]$ and $U=0,4,8$ (in units of hopping $J$). All other system variables are kept the same as defined in \cref{sec:dimer}. We would like to compare the actions of \cref{eqn:purely_fermion_lindbladian} and \cref{eq:Lind_eph} on the \gls{ldos} ~\cite{economou2006green}, an intrinsic quantity describing a system's electronic structure as well as information of the excitation spectra. The \gls{ldos} on site $i$ is given by 
\begin{equation}
    \centering
    \tilde{\rho}_{i}(\omega) := -\frac{1}{\pi} \text{Im }\mathcal{G}^{R,\sigma}_{ii}(\omega),
    \label{eqn:LDOS}
\end{equation}
and is attainable via the diagonal components of $\mathcal{G}^{R,\sigma}(\omega)$. This is the Fourier transform~\footnote{Formally, the Fourier transform is defined as 
\begin{equation*}
\tilde{G}^{R, \sigma}(\omega) := \int_{-\infty}^{\infty} G^{R,\sigma}(t) e^{-i\omega t} dt
\end{equation*} 
where the integral extends to infinity. Practically, we compute the dynamical correlation function $G_{ii}^{R,\sigma}(t)$ as a time truncated signal with some fixed sampling rate. Therefore, $\mathcal{G}_{ii}^{R, \sigma}(\omega)$ as defined in the text is an approximation of $\tilde{G}^{R, \sigma}(\omega)$. Thus, if $G_{ii}^{R,\sigma}(t)$ is represented by a time signal $x_1, \dots, x_M$ with $M$ points, then $\mathcal{G}_{ii}^{R,\sigma}(\omega)$ is computed through the \gls{fft}
\begin{equation*}
    \mathcal{G}_{ii}^{R,\sigma}(\omega) := \sum_{p=1}^{p=M} x_p e^{-\frac{i2\pi \omega p}{M}}
\end{equation*}} of the dynamical correlation function in time (or many-body Green's function) $G^{R,\sigma}(t)$ defined as
\begin{equation}
    \centering
    G_{ii}^{R,\sigma}(t) := -i\theta(t) \langle \{ c^{\dagger}_{i,\sigma}(t), c_{i,\sigma} \} \rangle_0 
    \label{eqn:retarded_GF}
\end{equation}
where the overlap $\langle \rangle_0$ is taken over the ground state. $\theta(t)$ is the Heaviside step function to impose causality. There is a global spin symmetry under exchanging $\uparrow$ and $\downarrow$ such that $G^{R, \uparrow}_{ii}(t) = G^{R, \downarrow}_{ii}(t)$. We therefore restrict ourselves to considering just the spin up sector and drop the spin index. Since the dimer is invariant under swapping of sites, the two sites have equivalent \gls{ldos} and we fix $i=1$. For the spinful interacting Fermi-Hubbard dimer without any dephasing, the zero-temperature equilibrium Greens function has an analytical solution in the Lehmann representation~\cite{romaniello2021hubbard}, whose poles can be obtained exactly. These correspond to the energies required for adding/removing an electron and are used as reference points in \cref{fig:dcube_dimer_GF_spinful}. The effect of dissipation captured by the evolution with Lindbladians \cref{eqn:purely_fermion_lindbladian} and \cref{eq:Lind_eph} on the 
\gls{ldos} of the dimer is shown in \Cref{fig:dcube_dimer_GF_spinful}. The \gls{ldos} spectrum is symmetric under $\omega \rightarrow -\omega + c$ for constant $c$ due to particle-hole symmetry. The time truncation window and number of samples affect the \gls{fft} to give the non uniform intensity profile shown in the first row of \cref{fig:dcube_dimer_GF_spinful}. For the second and third row, the dimer evolves subject to the Lindbladians \cref{eqn:purely_fermion_lindbladian} and \cref{eq:Lind_eph} respectively. In the second row, the `stochastic' simulation method corresponds to \cref{thm:1}, i.e. the expectation over stochastic channels. This is much cheaper than full statevector simulation as it reduces the Hilbert space from dimension $D$ to $\sqrt{D}$. The effect of dephasing on the time signal is an exponential dampening factor, and hence manifests in the \gls{fft} as a Lorentzian envelope around the underlying frequencies. This is clearly demonstrated in \cref{fig:dcube_dimer_GF_spinful}, whereby an increase in $g$ increases the Lorentzian width. This smearing effect also uniformises the peak intensities by softening the abrupt time truncation window. The main difference in the two dephasing Lindbladians on the dimer \gls{ldos} under the same coupling strength $g$ is the resulting Lorentzian width. In the third row, the exact Lindbladian simulation is again not included. A quantifiable error bound does not exist for this simulation due to the truncation of the bosonic Hilbert space. On the other hand, Dcube computes $G^{R,\sigma}_{ii}(t)$ in the infinite boson Hilbert space \textit{and} has the promise of being theoretically bounded by Trotter error. We provide a comparative study of `exact Lindbladian' vs Dcube simulation for $G_{ii}^{R,\sigma}(t)$ in \cref{appendix:gf_dcube_vs_exact_tds} at short times, where for sufficiently large number of Trotter steps, the main error is due to the bosonic truncation. This is similar to what has been done for fermion densities in \cref{fig:TDS_1x2_gpaper_4.0_10000_samples_and_density_errors_ALLMETHODS} and \cref{fig:Lindblad_exact_nxbosonscaling_1x2_gpaper_4.0}
. 

\begin{figure}[H]
    \centering
    \includegraphics[width=\linewidth]{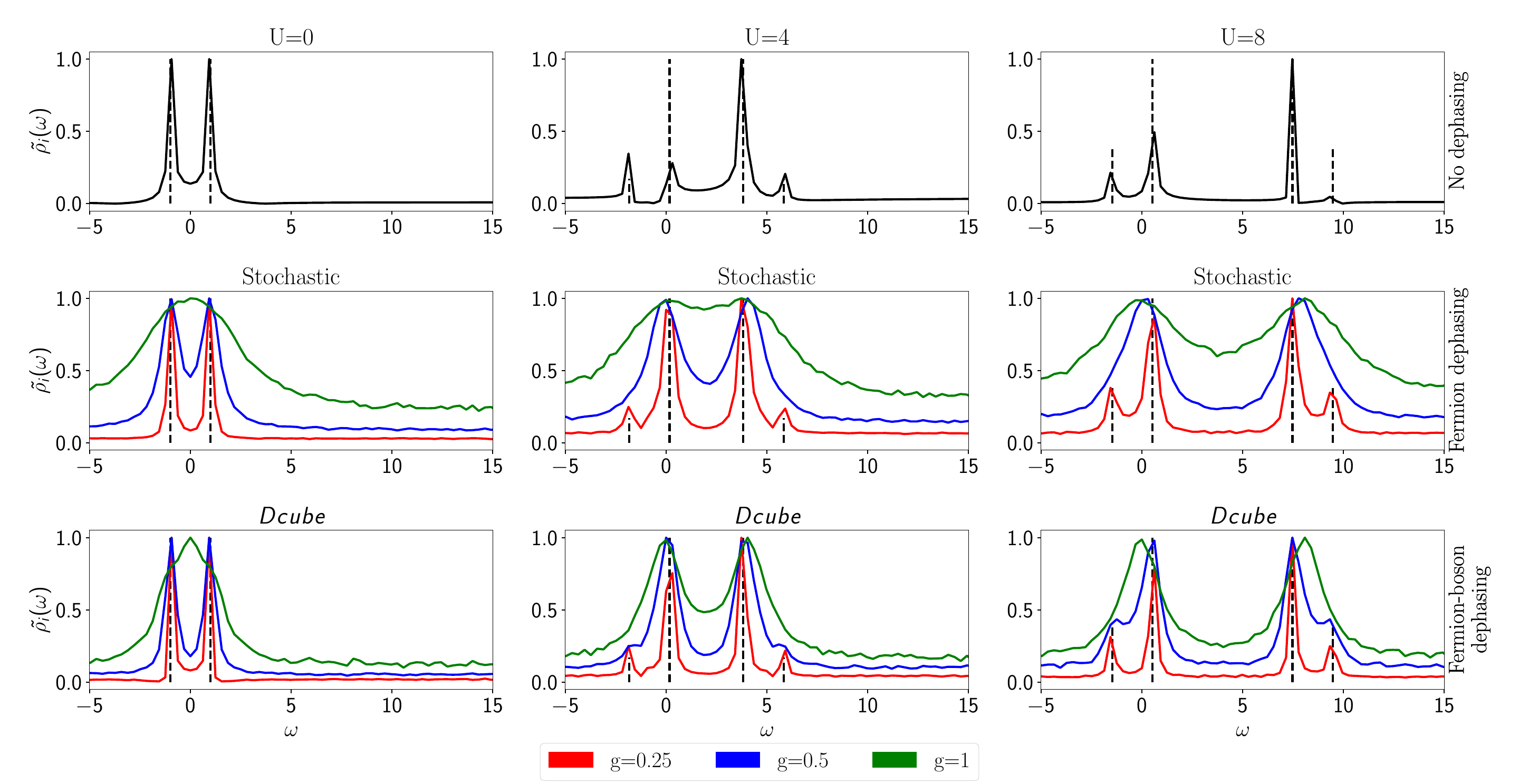}
    \caption{\gls{ldos} of the spinful Fermi-Hubbard dimer under two different types of dephasing at half-filling with $S_z=0$ at site $i=1$. Each column denotes a fixed interaction strength $U$. The vertical dashed lines denote the analytical peaks of the one-body Green's function for the dimer with no dephasing. Along row one, the solid black lines denote the rescaled imaginary component of the \gls{fft} of the truncated time signal $G_{11}(t)$.  Along the second and third row, the dimer is subject to the Lindbladians \cref{eqn:purely_fermion_lindbladian} and \cref{eq:Lind_eph} respectively. Different dephasing strengths $g$ are denoted by different colours. The method for computing the second row named `stochastic' is by producing the expectation value from $30$k stochastic channels through the application of \cref{thm:1} with $N=15$ first order Trotter steps. Dcube is used as the simulation method on the third row, with $N=10$ first order Trotter steps, $10$k IQP samples, $N_{\eta}=10$ and $N_{\rho}=2$ second order ($p=2$) Trotter steps to produce the map in \cref{eqn:Gkepsilonmap}. The total time evolution for all signals were performed up to $T=20$ with $100$ time point samples.}
    \label{fig:dcube_dimer_GF_spinful}
\end{figure}

\section{Conclusion}
\label{sec:conclusion}
The simulation of open quantum systems is fundamental to many aspects of physics, such as material science and condensed matter physics. Accurate simulations of open systems can quickly become intractable upon considering relevant environmental effects \cite{luchnikov2019simulation}. In this work we develope quantum algorithms capable of simulating Lindbladian dynamics (arguably the simplest type of evolution characterising open quantum systems).  We highlight two main contributions of this work, the first
is a framework for simulating general Lindbladians via sampling unitary channels. We provide an explicit algorithm to simulate unital Lindbladians without requiring additional qubits and a generalisation to simulate general Lindbladians using with  and additional ancillary qubits; the second introduces a
decoupling scheme that reduces the complexity of Lindbladians consisting of interacting subsystems. 

This decoupling scheme is particularly helpful when discussing systems coupled to bosonic degrees of freedom. Bosons have an infinite dimensional Hilbert space and then any finite dimensional representation introduces an error. Although in systems where the energy is conserved this error can be upper bounded \cite{woods2015simulating, tong2022provably, arzani2025effective}, estimating this error in practice may be difficult. Using the decoupling algorithm we find that the evolution of bosons coupled to fermions can be approximated by sampling an \gls{iqp} circuit, and using the sampled bitstring to define a fermionic circuit. This realises a unital channel on the fermions with a time dependent probability distribution, thus going beyond the Markovian approximation. 

We perform numerical simulations of a fermion-boson coupled system undergoing dephasing, which in itself consists of both fermionic and bosonic components. This illustrative example focuses on fermion dynamics, and brings attention to multiple issues surrounding the exact Lindbladian simulation involving bosons. One notable problem that we highlight is the issue of truncation of the bosonic Hilbert space. Although the simulation for a fixed truncation can appear to be converged and well behaved, increasing the truncation level can still affect the dynamics, signaling the error introduced by introducing an arbitrary cut-off.

The decoupling framework presented here is an accessible way of studying non-Markovian dynamics on a quantum computer and we envision several applications and generalisations to other practical scenarios. The connection with sampling \gls{iqp} circuits suggests that simulating unital Lindbladians including bosons could be classically difficult in general, even though the infinite temperature state is a fixed point of the dynamics.

\section{Acknowledgments}
\label{sec:acknowledgements}
We would like to thank Ashley Montanaro, Raul Garcia-Patron and Jan Lukas Bosse for helpful discussions relating to \gls{iqp} circuits,  as well as algorithmic details and applications.

%


\appendix

\section{Proof of \cref{thm:1}}\label{sec:proofs}

The aim of this section is to bound the distance between the mixed unitary channel $\mathcal{E}(\rho) = \sum_{{\bm s}=\{+,-\}^N}p_{\bm s}\mathcal{U}_{\bm s}(\rho)$, with $\mathcal{U}_{\bm s}$ defined in \cref{eq:mixed_unitary} and the evolution generated by $\mathcal{L}$. For simplicity we will be setting the coherent part to zero $(H=0)$ here as this is just captured by the usual first order Trotter error)

To this end we define two auxiliary channels
\begin{align}
    \Omega_{k,t}():=\sum_{\bm{s}=\{+,-\}^j}\frac{1}{2^k}\prod_{j=1}^ke^{is_j\sqrt{t}L_j}()e^{-is_j\sqrt{t}L_j}=\prod_j \omega_{j,t}()
\end{align}
and $e^{it\mathcal{L}_j}$, where $\mathcal{L}_j(\rho) = \sum_{i=1}^j L_i(\rho)L_i -\frac{1}{2}\{L_i^2,\rho\}$. Using these channels and the properties of the telescopic sum with $\mathcal{L}_{0}=0$ and $\Omega_{0,t}=1$
we have
\begin{align}
\left\|\Omega_{N,t}-e^{t\mathcal{L}_N}\right\|_\diamond=&\left\Vert \sum_{j=1}^{N}\Omega_{j,t}e^{t(\mathcal{L}_{N}-\mathcal{L}_j)}-\Omega_{j-1,t}e^{t(\mathcal{L}_{N}-\mathcal{L}_{j-1})}
\right\Vert _{\diamond}\\\leq&\sum_{j=1}^{N}\left\Vert \Omega_{j-1,t}(\omega_{j,t} - e^{t(\mathcal{L}_N-\mathcal{L}_{j-1})}e^{-t(\mathcal{L}_N-\mathcal{L}_{j})})e^{t(\mathcal{L}_N-\mathcal{L}_{j})}\right\Vert _{\diamond}
\end{align}
where the first equality is an identity of telescopic sums and the
last inequality is a consequence of the triangle inequality. Using
the submultiplicative
property of the diamond norm and the fact that the diamond norm of
a CPTP map is equal to one, the last inequality becomes
\begin{align}
\left\|\Omega_{N,t}-e^{t\mathcal{L}_N}\right\|_\diamond &\leq \sum_{j=1}^N\left\Vert \omega_{j,t}-e^{t(\mathcal{L}_N-\mathcal{L}_{j-1})}e^{-t(\mathcal{L}_N-\mathcal{L}_{j})}\right\Vert _{\diamond}\\
&\leq \sum_{j=1}^N\left\| \omega_{j,t}-e^{t(\mathcal{L}_j-\mathcal{L}_{j-1})}\right\|_\diamond + \sum_{j=1}^N\left\|e^{t(\mathcal{L}_j-\mathcal{L}_{j-1})}-e^{t(\mathcal{L}_N-\mathcal{L}_{j-1})}e^{-t(\mathcal{L}_N-\mathcal{L}_{j})}\right\|_\diamond
\end{align}
where ${\rm Adj}_A:=[A,]$ and ${\rm Adj}^n_A()=[{\rm Adj}_A^{n-1},()]$ the last inequality follows from the triangle inequality. The last term above is $O(t^2)$ by usual Trotter bounds, so we concentrate on the first one.

Let's define $\ell_j()=\mathcal{L}_j()-\mathcal{L}_{j-1}()$ and
\begin{align}
E_{j,t}:=\omega_{j,t}-e^{t(\mathcal{L}_{j}-\mathcal{L}_{j-1})}=	\sum_{m=0}^{\infty}t^{m}\left(\frac{(-1)^{m}}{(2m)!}{\rm Adj}_{L_{j}}^{2m}()-\frac{{\rm \ell}_j^{m}()}{m!}\right)
\end{align} now using that $ {\rm Adj}_{L_{j}}^{0}()=1$ and 
\begin{align}
    {\rm Adj}_{L_{j}}^{2}()=[L_{j},[L_{j},()]]=[L_{j},L_{j}()-()L_{j}]=\{L_{j}^{2},()\}-2L_{j}()L_{j}=-\ell_{j}
\end{align}
by induction it follows that ${\rm Adj}_{L_j}^{2m}=(-1)^m\ell_j^m$, so we have
\begin{align}E_{j,t}=\sum_{m=2}^{\infty}t^{m}\left(\frac{1}{(2m)!}-\frac{1}{m!}\right){\rm \ell}_{j}^{m}
\end{align}
finally, for $t\in[0,1], t^{2}\geq t^{m}$ so
\begin{align}
    \| e_{j,t}\|_\diamond\leq t^2\sum_{m=2}^\infty\left|\frac{1}{(2m)!}-\frac{1}{m!}\right|\|\ell_j^m\|_\diamond
\end{align}
which proves \cref{thm:1}.

\section{Electron phonon interacting Lindbladian derivation}
\label{appendix:epi_lindbladian_derivation}
Here we discuss the origins of the electron-phonon Lindbladian \cref{eq:Lind_eph} from the main text. We start from a closed system of $L$ sites whose Hamiltonian consists of fermions and bosons given by \cref{eqn:full_hamiltonian_before_trace_bath}. 
For concreteness, we write for the fermion subsystem $H_F$, a quartic Hamiltonian with creation (annihilation) operators denoted by $c^{\dagger}$ ($c$). This choice of $H_F$ can be arbitrary as it eventually will only play a role in the unitary evolution. The boson subsystem is further divided into Brownian and bath particles residing on each site $i$, whose creation (annihilation) operators are denoted by $q^{\dagger} (q)$ and $b^{\dagger} (b)$ respectively. 
\begin{align}
\centering
H_{T} &=  \Big[ \sum_{ ijkl} c^{\dagger}_i c_j + h.c. + V_{ijkl} c^{\dagger}_i c^{\dagger}_j c_c c_k  \Big] \text{\textcolor{blue}{($H_F$)}}+ \Big[ \sum_i {\omega}_i q^{\dagger}_i q_i \Big]  \text{\textcolor{blue}{Brownian particle ($H_Q$)}} \nonumber\\
& + \Big[ \sum_i \frac{g_i}{\sqrt{2\pi}} x_i X_{i,\alpha} \left(c^{\dagger}_i c_i-\frac{1}{2}\right) \Big]  \text{\textcolor{blue}{2nd order int ($H_{FQB}$)}} + \Big[ \sum_{i,\alpha} \tilde{\omega}_{i,\alpha} b^{\dagger}_{i,\alpha} b_{i,\alpha} \Big]  \text{\textcolor{blue}{Bath ($H_B$)}} 
\label{eqn:full_hamiltonian_before_trace_bath}
\end{align}
The fermions and bosons on site $i$ are coupled to each other with strength $\frac{g_i}{\sqrt{2\pi}}$, and $\tilde{\omega}_{i,\alpha}$ (${\omega}_i$) denote the bath (Brownian particle) harmonic oscillator frequencies respectively. $\alpha$ is a frequency index.
 $x_i := \frac{q_i^{\dagger}+q_i}{\sqrt{2}}$ and $X_{i,\alpha}:= \frac{b^{\dagger}_{i,\alpha} + b_{i,\alpha}}{\sqrt{2}}$ denote the position operators for the Brownian and bath particles respectively. From here on we work in the interaction picture of the bath, $H_B$. For simplicity, we renormalise such that $\hslash=1$ and assume that each site has the same coupling strength $\frac{g_i}{\sqrt{2\pi}} =\frac{g}{\sqrt{2\pi}}, \forall i$. The total Hamiltonian in the interaction picture is therefore
\begin{align}
\centering
H_I(t) &= H_F + H_Q + e^{-itH_B} H_{FQB}  e^{itH_B} \\
&= H_F+ H_Q + H_{FQB}(t).
\end{align}
Using the boson commutation relations, the interaction Hamiltonian between different subsystems simplifies to 
\begin{align}
\centering
H_{FQB}(t) &= \sum_{i,\alpha} \frac{g}{\sqrt{2\pi}} x_i \left(c^{\dagger}_ic_i -\frac{1}{2}\right) X_{i,\alpha}(t)
\end{align}
where $X_{i,\alpha}(t) = e^{-itH_B} X_{i,\alpha} e^{itH_B}$. Under the Born approximation \cite{born1926quantenmechanik} (i.e. weak coupling with bath) we approximate the total density matrix as $\rho_T(t) \approx  \rho_{FQ}(t) \otimes \rho_B$. Assuming that the bath is memoryless, we apply the Markov assumption \cite{redfield1957theory,davies1974markovian} and the \gls{bmme} is given by
\begin{equation}
\centering
\frac{d}{dt} \rho_{FQ}(t) = -i \text{tr}_B ([H_{I}(t), \rho_{FQ}(t) \otimes \rho_B]) -\int_0^{\infty} ds \text{ tr}_B \Big( [H_I(t), [H_I(t-s), \rho_{FQ}(t) \otimes \rho_B]] \Big)
\label{eqn:BMME_general}
\end{equation} 

We will refer to the first and second term of \cref{eqn:BMME_general} as the unitary $\mathcal{L}_U$ and dissipative $\mathcal{L}_D$ components. Consider the case where the bath is initially in the vacuum state such that 
\begin{equation}
    \centering
    \text{Tr}_B(\rho_B X_{i,\alpha}) = 0, \forall i, \alpha
    \label{eqn:trace_bath_zero}
\end{equation} 
Using \cref{eqn:trace_bath_zero}, the unitary evolution component of the \gls{bmme} in \cref{eqn:BMME_general} therefore simplifies to 
\begin{equation}
    \centering
    \mathcal{L}_U(\rho_{FQ}(t)) = -i ([H_F + H_Q, \rho_{FQ}(t)])\label{eqn:unitary_part_traced_bath}
\end{equation}
after tracing out the bath, where the $H_{FQB}(t)$ term vanish. For the dissipative component $\mathcal{L}_D$, four of the commutator terms involving a single $H_{FQB}(t)$ are again zero using \cref{eqn:trace_bath_zero}. If we integrate \cref{eqn:BMME_general} with respect to time $t$, at short times the $\mathcal{L}_U$ scales $O(t)$ and $\mathcal{L}_D$ scales as $O(t^2)$. The four terms that are time-independent coming from $\text{tr}_B([H_F + H_Q, [H_F + H_Q, \rho_{FQ}(t) \otimes \rho_B]])$ scale as second order in time, i.e. $O(t^2)$ and are ignored. Thus, there is only one non-trivial term out of the nine to consider, given by
\begin{align}
\centering
 \mathcal{L}_D (\rho_{FQ}(t))&= -\int_0^{\infty} ds \text{ tr}_B \Big([H_{FQB}(t), [H_{FQB}(t-s), \rho_{FQ}(t)  \otimes \rho_B]]
\label{eqn:BMME_4terms}
\end{align}
If we assume that the brownian boson is initialised in the vacuum state, we arrive at the simplified relations
\begin{align}
    \centering
    \int ds \frac{g^2}{2\pi} \text{ tr}_B \Big(\sum_{ij} X_i(t) X_j(t-s) \rho_B  \Big) & =   \int ds \frac{g^2}{2\pi} \text{ tr}_B \Big(\sum_{ij} X_i(t-s) \rho_B X_j(t)  \Big) = \sum_{i,\alpha} \frac{g^2}{2\pi} \Big( \int_0^{\infty} e^{-i\tilde{\omega}_{i,\alpha} s} ds \Big)  \nonumber \\
    \int ds \frac{g^2}{2\pi} \text{ tr}_B \Big(\sum_{ij}  \rho_B X_i(t-s) X_j(t) \Big) &=   \int ds \frac{g^2}{2\pi} \text{ tr}_B \Big(\sum_{ij} X_i(t) \rho_B X_j(t-s)  \Big) = \sum_{i,\alpha}  \frac{g^2}{2\pi} \Big( \int_0^{\infty} e^{i\tilde{\omega}_{i,\alpha} s} ds\Big) 
\label{eqn:trace_bosons_simplify} 
\end{align}
We also extend the integral to $-\infty$ and take the limit of a continuum of bath modes such that $\sum_{\alpha} \rightarrow \int d\tilde{\omega}_{\alpha}$ such that 
\begin{equation}
    \centering
    \int d\tilde{\omega}_i \int_{-\infty}^{\infty} ds e^{\pm i\tilde{\omega}_{i,\alpha}s}  = 2\pi \int d\tilde{\omega}_i  \delta(\tilde{\omega}_i) = 2\pi
\end{equation}

then \cref{eqn:BMME_4terms} simplifies to
\begin{align}
    \centering
    \frac{d}{dt} \rho_{FQ}(t)(t) &= -i[H_F + H_Q, \rho_{FQ}(t)] \nonumber \\
    &+ \sum_i g^2  \Big(x_i (n_i \rho_{FQ}(t) n_i - \frac{1}{2}\{n_i, \rho_{FQ}(t) \} \Big) x_i + \frac{1}{4} (x_i \rho_{FQ}(t) x_i - \frac{1}{2}\{x_i^2,\rho_{FQ}(t) \}) \Big)
\end{align}
which is exactly the form of the Lindbladian  \cref{eq:Lind_eph} studied in the main text. 

\section{Simulating dynamical correlation function time dynamics}
\label{appendix:gf_dcube_vs_exact_tds}

In this section, we plot the \gls{tds} of the dynamical correlation function $\text{Re}\left(\langle \{ c_1^{\dagger} c_1 (t)\} \rangle_0\right)$ using Dcube and `exact Lindbladian' methods. As shown in \cref{fig:dcube_dimer_exact_lindblad_vs_d3_scaling}, truncating the boson Hilbert space in `exact Lindbladian' simulations has immediate consequences at short times. This behavior was different for fermion densities studied in \cref{fig:Lindblad_exact_nxbosonscaling_1x2_gpaper_4.0}, which were more stable at short times. This highlights an important practical point, in that it is necessary to do a boson truncation study for \textit{every} observable of interest. On the other hand, the error of Dcube is independent of the truncation, and converges in this particular case at $N_t=5$ first order Trotter steps.

\begin{figure}[H]
    \centering
    \includegraphics[width=0.7\linewidth]{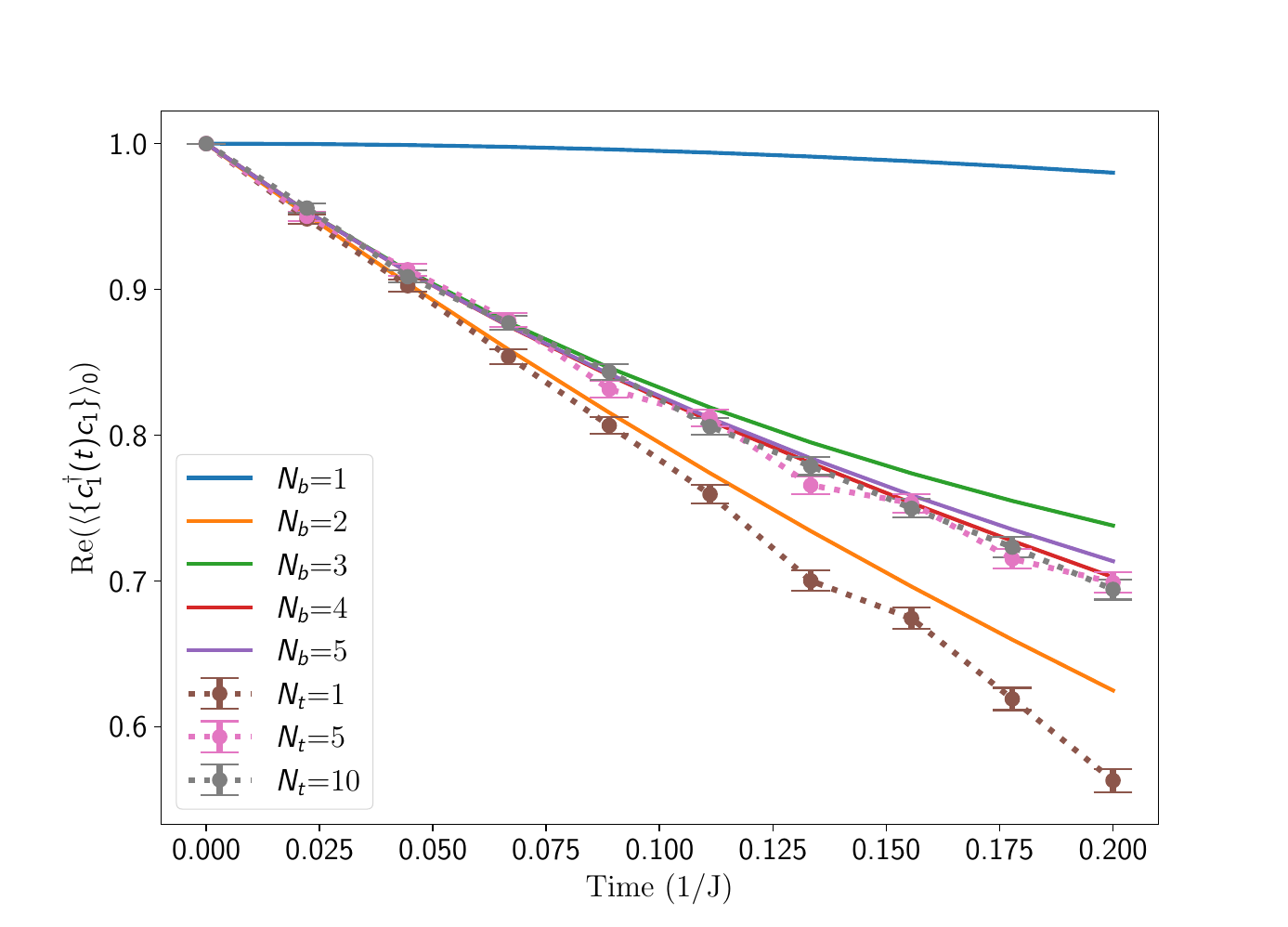}
    \caption{Time dynamics of $\text{Re}\left( \langle \{ c^{\dagger}_1(t) c_1 \}_0 \rangle \right)$ for the spinful dimer at half-filling at site one, $S_z=0$ and $g=1$. The solid lines denote exact Lindbladian statevector simulation with truncated Boson Hilbert space of with $N_b$ levels (i.e. with maximum $N_b-1$ bosons). The dashed lines denote using the Dcube algorithm (\cref{lem:boson_qubit} ) with
    $10$k IQP samples at various first order Trotter step sizes $N_t$. The error bars are taken as $\frac{\sigma}{\sqrt{N_{\text{s}}}}$ where $\sigma$ is the sample standard deviation.} \label{fig:dcube_dimer_exact_lindblad_vs_d3_scaling}
\end{figure}

\end{document}